%% file: FM2016-HCSPApprox.tex
\documentclass[runningheads,a4paper]{llncs}

\usepackage{latexsym}
\usepackage{amsmath,amsfonts,amssymb}
\usepackage{stmaryrd}
\usepackage{wasysym}
\usepackage{graphicx}
\usepackage{extarrows}
\usepackage{enumitem}
\usepackage{algorithm}
\usepackage{algorithmic}
\usepackage{multirow}
\usepackage{xcolor}

\allowdisplaybreaks[4]
\newcommand{\pskip}{\textmd{skip}}
\newcommand{\nstop}{\textbf{stop}}

\newcommand{\evolution}[3]{\langle {#1}(\dot{#2},#2)=0 \& #3\rangle}
\newcommand{\evolutionn}[2]{\langle #1 \& #2\rangle}

\newcommand{\pwait}{\textrm{wait}}
\newcommand{\exempt}[4]{#1 \unrhd \talloblong_{#2} (#3 \rightarrow #4)}

\newcommand*{\QEDB}{\hfill\ensuremath{\square}}
\newcommand{\discrete}[3]{\textit{D}_{#1, #2}(#3)}

\newcommand{\compress}[3]{#1\stackrel{#2}{\twoheadrightarrow}#3}
\newcommand{\fracN}[2]{\frac{\small \textstyle #1}{\small \textstyle #2}}
\newcommand{\bisimilar}[2]{#1\cong_{h,\varepsilon}#2}
\newcommand{\euler}{\xi}
\newcommand{\states}{\textit{postState}}

\newcommand{\RR}{\mathbb{R}}

\newcommand{\qq}{\mathbf{q}}
\newcommand{\s}{\mathbf{s}}
\newcommand{\pp}{\mathbf{p}}
\newcommand{\yy}{\mathbf{y}}

\newcommand{\xx}{\mathbf{x}}

\newcommand{\dd}{\mathbf{d}}

\newcommand{\ff}{\mathbf{f}}

\newcommand{\ggas}{{\small $\mathrm{GAS}$}}
\newcommand{\vare}{\varepsilon}
\newcommand{\labeld}{\textit{dis}}
\newcommand{\Define}{\stackrel{\mbox{\small\rm def}}{=}}
\newcommand{\vectorr}{\textit{vec}}
\newcommand{\first}{\textit{fst}}
\newcommand{\VAR}{\textit{Var}}
\newcommand{\VAL}{\textit{Val}}
\newcommand{\second}{\textit{snd}}
\newcommand{\repnum}{\textit{num}}
\newcommand{\subp}{\textit{subp}}
\newcommand{\oomit}[1]{}

\usepackage{wrapfig,lipsum,booktabs}
\usepackage{times}
\usepackage{bm}

\usepackage{url}
\urldef{\mailsa}\path|{yangg,ljiao,yangjia,wangsl,znj}@ios.ac.cn|
\newcommand{\keywords}[1]{\par\addvspace\baselineskip
\noindent\keywordname\enspace\ignorespaces#1}
\graphicspath{{figures/}}

\begin{document}
\mainmatter  % start of an individual contribution

\title{Approximate Bisimulation and Discretization of Hybrid CSP
\thanks{This work is supported partly by ``973 Program" under grant No. 2014CB340701, by NSFC under grants 91418204 and 61502467, by CDZ project CAP (GZ 1023), and by the CAS/SAFEA International Partnership Program for Creative Research Teams.}
\thanks{The corresponding authors: Shuling Wang (wangsl@ios.ac.cn) and Naijun Zhan (znj@ios.ac.cn).
}}

\titlerunning{Approximate Bisimulation and Discretization of Hybrid CSP}

\author{Gaogao Yan, Li Jiao, Yangjia Li, Shuling Wang and Naijun Zhan}
\authorrunning{G. Yan et al.}

\institute{State Key Lab. of Comput. Sci.,
Institute of Software, Chinese Academy of Sciences \\
\mailsa\\}


\toctitle{Lecture Notes in Computer Science}
\tocauthor{Authors' Instructions}
\maketitle
\pagestyle{empty}
\begin{abstract}
Hybrid Communicating Sequential Processes (HCSP) is a powerful formal modeling language for hybrid systems, which is an extension of CSP by introducing differential equations for modeling continuous evolution and interrupts for modeling interaction between continuous and discrete dynamics. In this paper, we investigate the semantic foundation for HCSP from an operational point of view by
proposing the notion of approximate bisimulation, which provides an appropriate criterion to characterize the equivalence
between HCSP processes with continuous and discrete behaviour. We give an algorithm to determine whether two HCSP processes are approximately bisimilar.
In addition, based on which, we propose an approach on how to discretize HCSP, i.e., given an HCSP process $A$, we construct another HCSP process $B$ which does not contain any continuous dynamics such that $A$ and $B$ are approximately bisimilar with given precisions.
This provides a rigorous way to transform a verified control model to a correct program model, which fills the gap in the design of embedded systems.

\keywords{HCSP, approximately bisimilar, hybrid systems, discretization}
\end{abstract}

%\begin{abstract}
%The abstract should summarize the contents of the paper and should
%contain at least 70 and at most 150 words. It should be written using the
%\emph{abstract} environment.
%\keywords{We would like to encourage you to list your keywords within
%the abstract section}
%\end{abstract}

%\input{chapter/C1-Introduction}
%\input{chapter/C2-Preliminaries}
%\input{chapter/C3-Transition-system-and-approximate-bisimulation}
%\input{chapter/C4-Transition-system-of-HCSP-processes}
%\input{chapter/C5-Discretization-of-HCSP}
%\input{chapter/C6-Case-study}
%\input{chapter/C7-Conclusion}

\input{C1-Introduction}
\input{C2-Preliminaries}

\input{C3-Transition-system-and-approximate-bisimulation}

\input{C4-Transition-system-of-HCSP-processes}

\input{C5-Discretization-of-HCSP}
\input{C6-Case-study}

\input{C7-Conclusion}

\bibliographystyle{plain}
\bibliography{springerbk}
\newpage
\input{Appendix}
\end{document}

%% file: C1-Introduction.tex
\section{Introduction}
Embedded Systems (ESs) make use of computer units to control physical processes so that the behavior of the controlled processes meets expected requirements. They  have become ubiquitous in our daily life, e.g., automotive, aerospace, consumer electronics, communications, medical, manufacturing and so on. ESs are used to carry out highly complex and often critical functions such as to monitor and control industrial plants, complex transportation equipments, communication infrastructure, etc. The development process of ESs is widely recognized as a highly complex and challenging task.
 %How to design correct embedded systems is a grand challenge for computer science and control theory.
 Model-Based Engineering (MBE) is considered as an effective way of developing correct complex ESs, and has been successfully applied in industry \cite{HS06,Lee00}. In the framework of MBE, a model of the system to be developed is defined at the beginning; then extensive analysis and verification are conducted based on the model so that errors can be detected and corrected at early stages of design of the system. Afterwards, model transformation techniques are applied to transform abstract formal models into more concrete models, even into source code.
 %Hybrid Systems (HSs) are mathematical models with precise mathematical semantics for ESs, wherein continuous physical dynamics are %combined with discrete transitions. Based on HSs, rigorous analysis and verification of ESs become feasible.

To improve the efficiency and reliability of MBE, it is absolutely necessary to automate the system design process as much as possible.
 This requires that all models at different abstraction levels have a precise mathematical semantics. Transformation between models at different abstraction levels should preserve semantics, which can be done automatically with tool support.

Thus, the first challenge in model-based formal design of ESs is to have a powerful modelling language which can model all kinds of features of ESs such as communication, synchronization, concurrency, continuous and discrete dynamics and their interaction, real-time, and so on, in an easy way. To address this issue, Hybrid Communicating Sequential Processes (HCSP) was proposed in \cite{He94,Zhou96},
  which is an extension of CSP by introducing differential equations for modeling continuous evolutions and interrupts for modeling interaction between continuous and discrete dynamics. Comparing with other formalisms, e.g., hybrid automata \cite{Henzinger96}, hybrid programs \cite{Platzer10}, etc., HCSP is more expressive and much easier to be used, as it provides a rich set of constructors. Through which
   a complicated ES with different behaviours can be easily modeled in a compositional way.
  The semantic foundation of HCSP has been investigated in the literature, e.g.,  in He's original work on HCSP \cite{He94},
an algebraic semantics of HCSP was given by defining a set of algebraic laws for
the constructors of HCSP. Subsequently, a DC-based semantics for HCSP was presented in
  \cite{Zhou96} due to Zhou \emph{et al}. These two original formal semantics of HCSP are very
  restrictive and incomplete, for example, it is unclear whether the set of algebraic rules defined in \cite{He94} is complete,
    and  super-dense computation and recursion are not well handled in \cite{Zhou96}.
  In \cite{LLQZ10,WZG12,Dimitar12,ZWZ13,UTP16}, the axiomatic, operational, and the DC-based and UTP-based denotational semantics for
  HCSP are proposed, and the relations among them are discussed.
  %In \cite{LLQZ10,ZWZ13,WZG12,GWZ13,UTP16}, operational, axiomatic and DC-based and UTP-based denotational semantics for
%  HCSP are proposed, and the relations among them are discussed.
However, regarding operational semantics, just a set of transition rules
  was proposed in \cite{ZWZ13}. It is unclear in what sense two HCSP processes are equivalent from an operational point of view,
  which is the cornerstone of operational semantics, also the basis of refinement theory for a process algebra. So, it absolutely deserves to investigate the semantic foundation of HCSP
  from an operational point of view.

Another challenge in the model-based formal design of ESs is how to transform higher level abstract models (control models) to
lower level program models (algorithm models), even to C code, seamlessly in a rigorous way. Although huge volume of  model-based development
approaches targeting embedded systems has been
 proposed  and used in industry and academia, e.g., Simulink/Stateflow \cite{slusing,sfusing}, SCADE \cite{scade6}, Modelica \cite{modelica}, SysML \cite{sysml}, MARTE \cite{marte},  %Metropolis \cite{metropolis},
 Ptolemy \cite{ptolemy},
  hybrid automata \cite{Henzinger96}, CHARON \cite{Alur03}, HCSP \cite{He94,Zhou96}, Differential Dynamic Logic \cite{Platzer10}, and Hybrid Hoare Logic \cite{LLQZ10}, the gap between higher-level control models and lower-level algorithm models still remains.

Approximate bisimulation \cite{Girard07a} is a popular method for analyzing and verifying complex hybrid systems. Instead of requiring observational behaviors of two systems to be exactly identical, it allows errors but requires the ``distance'' between two systems remain bounded by some precisions. In \cite{Girard08}, with the use of simulation functions, a characterization of approximate simulation relations between hybrid systems is developed. A new approximate bisimulation relation with two parameters as precisions, which is very similar to the notion defined in this paper, is introduced in \cite{Julius09}. For control systems with inputs, the method for constructing a symbolic model which is approximately bisimilar with the original continuous system is studied in \cite{Pola08}. Moreover, \cite{Majumdar12} discusses the problem for building an approximately bisimilar symbolic model of a digital control system. Also, there are some works  on building symbolic models for networks of control systems \cite{Pola14}. But for all the above works,
either discrete dynamics is not considered, or it is assumed to be atomic actions independent of the continuous variables. In \cite{HHWT98,Tiwaro08,Tomo2016}, the abstraction of hybrid automata is considered, but it is only guaranteed that the abstract system is an approximate simulation of the original system. In~\cite{Platzer12}, a discretization of hybrid programs is presented for a proof-theoretical purpose, i.e., it aims to have a sound and complete axiomatization relative to properties of discrete programs.
Differently from all the above works,  we aim to have a discretization of HCSP,  for which discrete and continuous dynamics, communications, and so on, are entangled with each other tightly, to guarantee that the discretized process has the approximate equivalence with the original process.
% This complicates the theory of  approximate bisimulation for HCSP very much.

%In this paper, we investigate the above mentioned issues and give our answers.
The main contributions of this paper include:
\begin{itemize}
\item First of all, we propose the notion of approximate bisimulation, which provides a criterion to characterize in what sense
two HCSP processes with differential kinds of behaviours are equivalent from an operational point of view. Based on which, a refinement theory for HCSP could be developed.
\item Then, we show that whether two HCSP processes are approximately bisimilar or not is decidable if all ordinary differential equations (ODEs) occurring in
      them satisfy globally asymptotical stability (GAS) condition (the definition will be given later). This is achieved by
      proposing an algorithm to compute an approximate bisimulation relation for the two HCSP processes.
\item Most importantly, we present  how to discretize an HCSP process (a control model) by a discrete HCSP process (an algorithm model), and prove they are approximately bisimilar,
    if the original HCSP process satisfies the GAS condition and is robustly safe with respect to some given precisions.
\end{itemize}

The rest of this paper is organized as follows: In Sec.~\ref{section:Preliminaries}, we introduce some preliminary notions on dynamical systems. Sec.~\ref{section:tsandappbisim}
defines transition systems and the approximate bisimulation relation between transition systems. The syntax and the transition semantics of HCSP, and the approximately bisimilar of HCSP processes are presented in  Sec.~\ref{section:tsofhcsp}. The discretization of HCSP is presented in Sec.~\ref{section:discretizationofhcsp}. Throughout the paper, and in Sec.~\ref{section:casestudy}, a case study on the water tank system \cite{Ahmad14} is shown to illustrate our method.  At the end, Sec.~\ref{section:conclusion} concludes the paper and discusses the future work. For space limitation, the proofs for all the lemmas and theorems are omitted, but can be found in ~\cite{reportpaper}.

%% file: C2-Preliminaries.tex
\section{Preliminary}
\label{section:Preliminaries}
In this section, we briefly review some notions in dynamical systems,  that can be found at \cite{khalil1996,sontag2013}.
%\subsection{Dynamical Systems}
%\comment{to add: dynamic system $\dot{x} = f(x)$, existence and uniqueness of trajectory, forward complete, $\delta$-\ggas, globally %asymptotic stable (\ggas), ($\delta$ -)Lyapunov functions, and so on. and also the related theorems, such as $\delta$-\ggas $\Rightarrow$ \ggas}
 In what follows, $\mathbb{N}$, $\mathbb{R}$, $\mathbb{R}^+$, $\mathbb{R}^+_0$ denote the natural, real, positive and nonnegative real numbers, respectively. Given a vector $\xx \in \mathbb{R}^n$, $\| \xx\|$ denotes the infinity norm of $\xx\in \mathbb{R}^n$, i.e., $\| \xx \|=\max\{|x_1|,|x_2|,...,|x_n|\}$. A continuous function $\gamma:\mathbb{R}^+_0 \to \mathbb{R}^+_0$, is said in class $\mathcal{K}$ if it is strictly increasing and $\gamma(0)=0$; $\gamma$ is
said in class $\mathcal{K}_\infty$ if $\gamma \in \mathcal{K}$ and $\gamma(r) \to \infty$ as $r \to \infty$. A continuous function $\beta:\mathbb{R}^+_0 \times \mathbb{R}^+_0 \to \mathbb{R}^+_0$ is said in class $\mathcal{K\!L}$ if for each fixed $s$, the map $\beta(r,s) \in \mathcal{K}_\infty$ with respect to $r$ and, for each fixed $r$, $\beta(r,s)$ is decreasing with respect to $s$ and $\beta(r,s) \to 0$ as $s \to \infty$.

A dynamical system is of the following form
\begin{eqnarray}
  \dot{\xx}=\ff(\xx), \ \ \ \ \xx(t_0)=\xx_0 \label{eq:dynamical}
\end{eqnarray}
where $\xx\in \mathbb{R}^n$ is the state and $\xx(t_0)=\xx_0$ is the \emph{initial condition}.
% and $u\in U_{in} \subseteq \mathbb{R}^m$ is the input. By input we mean any measurable, locally essentially bounded
%functions of time from intervals of the form $(a,b)\subseteq \mathbb{R}$ to $U$ with $b>a$ and we denote the set of such functions by $\mathcal{U}$.

%In the following, we introduce the notions related to dynamic systems, which can be found at \cite{khalil1996,sontag2013}.
Suppose $a<t_0<b$. A function $X(.):(a,b)\to \mathbb{R}^n$ is said to be a \emph{trajectory} (solution) of (\ref{eq:dynamical}) on $(a,b)$, if $X(t_0)=\xx_0$ and $\dot{X}(t)=\ff(X(t))$ for all $t \ge t_0$. In order to ensure the existence and uniqueness of trajectories, we assume $\ff$ satisfying the local Lipschitz condition, i.e., for every compact set $S\subset \mathbb{R}^n$, there exists a constant $L>0$ s.t.  $\| \ff(\xx)-\ff(\yy)\| \le L\| \xx-\yy\|$, for all $\xx,\yy \in S$. Then, we write $X(t,\xx_0)$ to denote the point reached at time $t \in (a,b)$ from initial condition $\xx_0$, which should be uniquely determined. In addition, we assume (\ref{eq:dynamical}) is \emph{forward complete} \cite{Angeli99}, i.e., it is solvable on an open interval $(a,+\infty)$. An equilibrium point of (\ref{eq:dynamical}) is a point $\bar{\xx}\in \mathbb{R}^n$ s.t.  $\ff(\bar{\xx})=0$.

%In this paper, several stability assumptions of the dynamical systems are needed, which we briefly recall in the followings.

\begin{definition} \label{definition:GAS}
A dynamical system of form (\ref{eq:dynamical}) is said to be \emph{globally asymptotically stable} ({\small $\mathrm{GAS}$}) if there exists a point $\xx_0$ and a function $\beta$ of class $\mathcal{K\!L}$ s.t.
\[
\begin{array}{l}
\forall \xx\in \mathbb{R}^n \ \forall t \ge 0. \| X(t,\xx)-\xx_0 \| \le \beta(\|\xx-\xx_0\|,t).
\end{array}
\]
\end{definition}
It is easy to see that the point $\xx_0$ is actually the unique equilibrium point of the system. When this point is previously known or can be easily computed, one can prove the system to be GAS by constructing a corresponding Lyapunov function. However, $\xx_0$ cannot be found sometimes, for example, when the dynamics $\ff$ of the system depends on external inputs and thus is not completely known. The concept of $\delta$-\ggas\ would be useful in this case.
%The definition of $\delta$-\ggas can be thought of as an incremental version of \ggas:
\begin{definition} \label{definition:deltaGAS}
A dynamical system of (\ref{eq:dynamical})  is said to be incrementally globally asymptotically stable ($\delta$-\ggas) if it is forward complete and there is a $\mathcal{K\!L}$ function $\beta$ s.t.
%for any $t \ge 0$, and any $\xx,\yy\in \mathbb{R}^n$ the following condition is satisfied:
\[
\begin{array}{l}
\forall \xx\in \mathbb{R}^n \ \forall \yy\in \mathbb{R}^n \ \forall t \ge 0.  \| X(t,\xx)-X(t,\yy) \| \le \beta(\| \xx-\yy \|,t).
\end{array}
\]
\end{definition}

In \cite{Angeli02},  the relationship between {\ggas} and {$\delta$-\ggas} was established, restated by the following proposition.  %establish , which can be deduced or found the proof in .
\begin{proposition} \label{proposition:GAS}
\begin{itemize}
 \item If (\ref{eq:dynamical}) is {$\delta$-\ggas}, then it is {\ggas}.
 \item If there exist two strictly positive reals $M$ and $\varepsilon$, and a differentiable function $V(\xx,\yy)$ with $\alpha_1(\|\xx-\yy\|) \le V(\xx,\yy) \le \alpha_2(\|\xx-\yy\|)$ for some $\alpha_1$, $\alpha_2$ and $\rho$ of class $\mathcal{K}_\infty$, s.t.
\[
\forall \xx,\yy\in \RR^n. \left( \begin{array}{cl}
& \|\xx-\yy\| \le \varepsilon \wedge
 \|\xx\| \ge M \wedge \|\yy\| \ge M \\
 \Rightarrow & \frac{\partial V}{\partial \xx}\ff(\xx) + \frac{\partial V}{\partial \yy}\ff(\yy) \le -\rho (\|\xx-\yy \|)
\end{array} \right),
\]
then the system (\ref{eq:dynamical}) is  $\delta$-\ggas.
\end{itemize}
\end{proposition}

 A function $V(\xx,\yy)$ satisfying the condition in Proposition \ref{proposition:GAS} is called a \emph{ {$\delta$-\ggas} Lyapunov function} of (\ref{eq:dynamical}). Proposition~\ref{proposition:GAS} tells us that (\ref{eq:dynamical}) is {$\delta$-\ggas} if and only if it admits a {$\delta$-\ggas} Lyapunov function. In general, checking the inequality in Def. \ref{definition:deltaGAS} is difficult, one may construct
  {$\delta$-\ggas} Lyapunov functions as an alternative.

%% file: C3-Transition-system-and-approximate-bisimulation.tex
\section{Transition systems and approximate bisimulation}
\label{section:tsandappbisim}
%\subsection{Transition systems}
In the following, the set of actions, denoted by $\textit{Act}$, is assumed to consist of
 a set of discrete actions which take no time (written as $\mathcal{E}$), $\mathbb{R}^+_0$ the set of
  delay actions which just take time delay, and a special internal action $\tau$. Actions are
   ranged over $l_1,\ldots, l_n, \ldots$.
\begin{definition}[Transition system] A labeled transition system with observations is a tuple $T=\langle Q,L,\to,Q^{0},Y,H\rangle $, where $Q$ is a  set of states, $L\subseteq \textit{Act}$ is a  set of labels, $Q^{0}\subseteq Q$ is a   set of initial states,
  $Y$ is a  set of observations,
and $H$ is an observation function $H:Q\to Y$,
$\to\subseteq Q\times L\times Q$ is a transition relation,   satisfying \\
 % \begin{enumerate}
   %\item
   1, \textbf{identity:} $q\overset{0}{\longrightarrow} q$ always holds; \\
   %\item
  2, \textbf{delay determinism:} if $q\overset{d}{\longrightarrow} q'$ and
   $q\overset{d}{\longrightarrow} q''$, then $q'=q''$; and \\
 % \item
  3, \textbf{delay additivity:}
     if $q\overset{d_{1}}{\longrightarrow} q'$ and $q'\overset{d_{2}}{\longrightarrow} q''$ then $q\overset{d_{1}+d_{2}}{\longrightarrow} q''$, where $d,d_1,d_2\in \mathbb{R}^+_0$.
  %\end{enumerate}
\label{definition:TS}
\end{definition}

%We assume a special label $\tau$, representing the silent actions that take no time to complete. Thus the actual set of labels is $L\cup %\{\tau\}$.
A transition system $T$ is said to be \emph{symbolic} if $Q$ and $L\cap \mathcal{E}$ are  finite, and $L\cap \mathbb{R}^+_0$ is bounded,
  and \emph{metric} if the output set $Y$ is equipped with a metric $\textbf{d}:Y\times Y\to \mathbb{R}^{+}_{0}$. In this paper, we regard $Y$ as being equipped with the metric $\textbf{d}(\yy_1,\yy_2)=\| \yy_1-\yy_2 \|$.

A  \emph{state  trajectory} of a transition system $T$ is a (possibly infinite) sequence of transitions $\qq^0 \xrightarrow{l^{0}} \qq^1 \xrightarrow{l^{1}}\cdots \xrightarrow{l^{i-1}} \qq^{i}\xrightarrow{l^{i}} \cdots$, denoted by $\{\qq^i\xrightarrow{l^i}\qq^{i+1}\}_{i\in \mathbb{N}}$,  s.t.  $\qq^{0}\in Q^0$ and for any $i$,  $\qq^{i} \xrightarrow{l^{i}} \qq^{i+1}$.  An \emph{observation trajectory} is a (possibly infinite) sequence $\yy^0 \xrightarrow{l^{0}} \yy^1 \xrightarrow{l^{1}}\cdots \xrightarrow{l^{i-1}} \yy^{i}\xrightarrow{l^{i}} \cdots$, denoted by $\{\yy^i\xrightarrow{l^i}\yy^{i+1}\}_{i\in \mathbb{N}}$, and it is accepted by $T$ if there exists a corresponding state trajectory of $T$ s.t.  $\yy^{i}=H(\qq^{i})$ for any $i\in \mathbb{N}$. The set of observation trajectories accepted by $T$ is called the \emph{language} of $T$, and is denoted by $L(T)$. The reachable set of $T$ is a subset of $Y$ defined by
  \[
 % \begin{array}{l}
\textit{Reach}(T)=\{\yy\in Y | \exists \{\yy^i\xrightarrow{l^i}\yy^{i+1}\}_{i\in \mathbb{N}}\in L(T),\exists j\in \mathbb{N},\yy^j=\yy  \}.
%  \end{array}
  \]
We can verify the safety property of $T$ by computing $\textit{Reach}(T)\cap Y_{U}$, in which $Y_{U}\subseteq Y$ is the set of unsafe observations. If it is empty, then $T$ is \emph{safe}, otherwise, \emph{unsafe}.

 For a maximum sequence of $\tau$ actions $\qq^{i} \xrightarrow{\tau} \qq^{i+1} \xrightarrow{\tau} \cdots \xrightarrow{\tau}\qq^{i+k}$,
% i.e. neither the previous action before $\qq^{i}$ nor the next action after $\qq^{i+k}$ is $\tau$,
we remove the intermediate states and define the \emph{$\tau$-compressed}  transition $\compress{\qq^{i}}{\tau}{\qq^{i+k}}$ instead. For unification, for a non-$\tau$ transition $\qq^{i}\xrightarrow{l^i}\qq^{i+1}$ where $l^i \neq \tau$, we define $\compress{\qq^{i}}{l^i}{\qq^{i+1}}$.
In what follows, we will denote $\langle Q,L,\twoheadrightarrow,Q^{0},Y,H\rangle $ the resulting labeled transition system from
  $\langle Q,L,\to,Q^{0},Y,H\rangle$ by replacing each label transition with its $\tau$-compressed version.
%As a result, for any state trajectory $\mathcal{S}$ of a transition system $T$, it can be transformed to a unique $\tau$-compressed trajectory, denoted by $\{\compress{\qq^{i}}{l^i}{\qq^{i+1}}\}_{i \in \mathbb{N}}$ where $l^i \in L\cup \{\tau\}$.
As a common convention in process algebra, we use  $\pp \xLongrightarrow{l}\pp^{\prime}$ to denote the closure of $\tau$ transitions, i.e., $\compress{\pp(}{\tau}{)^{\{0,1\}}} \compress{}{l}{(} \compress{}{\tau}{)^{\{0,1\}}\pp^{\prime}}$,  for any $l\in L$  in the sequel.

  % A transition system is $\tau$-compressed, if all state trajectories of it are $\tau$-compressed. The definition of reachable set of $\tau$-compressed transition systems can be modified correspondingly.  In the following,  we assume  all transition systems are compressed.

\oomit{Given two transition systems, we will define the approximate bisimulation relation between them by comparing the $\tau$-compressed trajectories of them respectively. }

Given $l_1, l_2\in L\cup \{\tau\}$, we define the \emph{distance} $\labeld (l_1, l_2)$ between them as follows:
\[\labeld(l_1, l_2) \Define \left\{
\begin{array}{lll}
   0  & \mbox{if both $l_1$ and $l_2$ are in $\mathcal{E}$ or are $\tau$ }\\
   |d-d'| & \mbox{if $l_1=d$ and $l_2=d'$ are both delay actions, i.e., $d, d'\in \mathbb{R}^+_0$}\\
     \infty & \mbox{Otherwise}
\end{array}
\right.\]

\begin{definition}[Approximate bisimulation] Let $T_{i}=\langle Q_i,L_i,\twoheadrightarrow_i,Q^{0}_i,Y_i,H_i\rangle $, $(i=1,2)$ be two metric transition systems with  the same output set $Y$ and metric $\textbf{d}$. Let $h$ and $\varepsilon$  be the time and value precision respectively.  A relation $\mathcal{B}_{h, \varepsilon}\subseteq Q_{1}\times Q_{2}$ is called a $(h, \varepsilon)$-approximate bisimulation relation between $T_{1}$ and $T_{2}$, if for all $(\qq_{1},\qq_{2})\in \mathcal{B}_{h,\varepsilon}$, \\
%1. $\qq_{1}(t)=\qq_{2}(t)$ and $\textbf{d}(H_{1}(\qq_{1}),H_{2}(\qq_{2}))\le \varepsilon$,\\
1. $\dd(H_{1}(\qq_{1}),H_{2}(\qq_{2}))\le \varepsilon$,\\
%2. $\forall \qq_{1}\xrightarrow{l}_{1} \qq^{\prime}_{1}$, $\exists \qq_{2}\xrightarrow{l}_{2} \qq^{\prime}_{2}$ s.t.  $(\qq^{\prime}_{1},\qq^{\prime}_{2})\in \mathcal{B}_{h,\varepsilon}$,\\
%3. $\forall \qq_{2}\xrightarrow{l}_{2} \qq^{\prime}_{2}$, $\exists \qq_{1}\xrightarrow{l}_{1} \qq^{\prime}_{1}$ s.t.  $(\qq^{\prime}_{1},\qq^{\prime}_{2})\in \mathcal{B}_{h,\varepsilon}$, where $l\in L\cup \{\tau\}$.
2. $\compress{\forall \qq_{1}}{l}{_1 \qq^{\prime}_{1}}$, $\exists \qq_{2}\xLongrightarrow{l'}_{2} \qq^{\prime}_{2}$ s.t.  $\labeld(l, l') \le h$ and  $(\qq^{\prime}_{1},\qq^{\prime}_{2})\in \mathcal{B}_{h, \varepsilon}$,   for $l\in L_1$ and $l'\in L_2$\\
3. $\compress{\forall \qq_{2}}{l}{_2 \qq^{\prime}_{2}}$, $\exists \qq_{1}\xLongrightarrow{l'}_{1} \qq^{\prime}_{1}$ s.t. $\labeld(l, l') \le h$ and  $(\qq^{\prime}_{1},\qq^{\prime}_{2})\in \mathcal{B}_{h, \varepsilon}$,  for $l\in L_2$ and $l'\in L_1$.
%where $\widehat{l}$ is the sequence gained by deleting all occurrences of $\tau$ from $l$.
\label{definition:appbisimulation}
\end{definition}
\begin{definition}$T_{1}$ and $T_{2}$ are approximately bisimilar with the precision $h$ and $\varepsilon$ (denoted $T_{1}\cong_{h, \varepsilon}T_{2}$), if there exists  a $(h, \varepsilon)$-approximate bisimulation relation $\mathcal{B}_{h,\varepsilon}$ between $T_{1}$ and $T_{2}$ s.t.  for all $\qq_{1}\in Q^{0}_{1}$, there exists $\qq_{2}\in Q^{0}_{2}$ s.t.  $(\qq_{1},\qq_{2})\in \mathcal{B}_{h,\varepsilon}$, and vice versa.
\label{definition:TSappbisimulation}
\end{definition}

The following result ensures that the set of $(h,\varepsilon)$-approximate bisimulation relations has a maximal element.

\begin{lemma}
Let$\{\mathcal{B}^i_{h, \varepsilon}\}_{i\in I}$ be a family of $(h,\varepsilon)$-approximate bisimulation relations between $T_1$ and $T_2$. Then, $\bigcup_{i\in I}\mathcal{B}^i_{h, \varepsilon}$ is a $(h,\varepsilon)$-approximate bisimulation relation between $T_1$ and $T_2$.
\label{lemma:union}
\end{lemma}
%\begin{proof}
%  For any $(\qq_1,\qq_2)\in \bigcup_{i\in I}\mathcal{B}^i_{h, \varepsilon}$, there exists $i\in I$ s.t.  $(\qq_1,\qq_2)\in \mathcal{B}^i_{h, \varepsilon}$. Then, $\dd(H_{1}(\qq_{1}),H_{2}(\qq_{2}))\le \varepsilon$. Moreover, for all $\compress{\qq_{1}}{l}{_1 \qq^{\prime}_{1}}$, there exists $\qq_{2}\xLongrightarrow{l^{\prime}}_{2} \qq^{\prime}_{2}$ s.t.  $(\qq^{\prime}_1,\qq^{\prime}_2)\in \mathcal{B}^i_{h, \varepsilon} \subseteq \bigcup_{i\in I}\mathcal{B}^i_{h, \varepsilon}$ and $\labeld(l, l') \le h$, and for all $\compress{\qq_{2}}{l}{_2 \qq^{\prime}_{2}}$, there exists $\qq_{1}\xLongrightarrow{l^{\prime}}_{1} \qq^{\prime}_{1}$ s.t.  $(\qq^{\prime}_1,\qq^{\prime}_2)\in \mathcal{B}^i_{h, \varepsilon} \subseteq \bigcup_{i\in I}\mathcal{B}^i_{h, \varepsilon}$ and $\labeld(l, l') \le h$. Therefore, $\bigcup_{i\in I}\mathcal{B}^i_{h, \varepsilon}$ is also a $(h,\varepsilon)$-approximate bisimulation relations between $T_1$ and $T_2$.\QEDB
%\end{proof}
By Lemma \ref{lemma:union}, given the precision parameters $h$ and $\varepsilon$,
%\begin{definition}
let $\{\mathcal{B}^i_{h, \varepsilon}\}_{i\in I}$ be the set of all $(h,\varepsilon)$-approximate bisimulation relations between $T_1$ and $T_2$, then the maximal $(h,\varepsilon)$-approximate bisimulation relation between $T_1$ and $T_2$ is defined by $\mathcal{B}^{max}_{h, \varepsilon}=\bigcup \limits_{i\in I}\mathcal{B}^i_{h, \varepsilon}$.
%%\label{definition:maxappbisimulation}
%%\end{definition}
%
%\oomit{Clearly, $T_1$ and $T_2$ are approximately bisimilar with precision $h$ and $\varepsilon$ if and only if for all $\qq_1\in Q^0_1$, there exists $\qq_2\in Q^0_2$ s.t.  $(\qq_1,\qq_2)\in \mathcal{B}^{max}_{h, \varepsilon}$, and vice versa.}
%
For two  transition systems that are approximately bisimilar, the reachable sets have the following relationship:
\begin{theorem}
If $T_{1}\cong_{h, \varepsilon}T_{2}$, then
%$
%\textit{Reach}(T_1)\subseteq N(cl(\textit{Reach}(T_2)),\varepsilon)
%$,
%where $cl(\textit{Reach}(T_2))$ means the closure of $\textit{Reach}(T_2)$ and $N(\yy,\varepsilon)$ denotes the $\varepsilon$ neighborhood of $\yy\in Y$.
$
\textit{Reach}(T_1)\subseteq N(\textit{Reach}(T_2),\varepsilon)
$,
where $N(Y,\varepsilon)$ denotes the $\varepsilon$ neighborhood of $Y$, i.e. $\{x \mid \exists y. y \in Y \wedge \| x-y \| < \varepsilon\}$.
\label{theorem:reachsetrelation}
\end{theorem}
Thus, if the distance between $\textit{Reach}(T_2)$ and the unsafe set $Y_U$ is greater than $\varepsilon$, then the intersection of $\textit{Reach}(T_1)$ and $Y_U$ is empty and hence $T_1$ is safe, whenever $T_{1}\cong_{h, \varepsilon}T_{2}$.

%% file: C4-Transition-system-of-HCSP-processes.tex
\section{Hybrid CSP (HCSP)}
\label{section:tsofhcsp}

In this section, we present a brief introduction to HCSP  and define the  transition system of HCSP  from an operational point of view. An example is given for better understanding. Finally, we investigate the approximate bisimilarity for HCSP processes.

 \subsection{HCSP}
Hybrid Communicating Sequential Process (HCSP) is a formal language for describing hybrid systems, which extends CSP by introducing differential equations for modelling continuous evolutions and interrupts for modeling the arbitrary interaction between continuous evolutions and discrete jumps. The syntax of HCSP can be described as follows:

%\begin{equation*}
% \begin{split}
%  P::=&skip|x:=e|wait\ d|ch?x|ch!e|P;Q|B\to P|P\sqcup Q|P^{\ast}\\
%  &|\langle F(\dot{s},s)=0\& B\rangle|\langle F(\dot{s},s)=0\& B \rangle \trianglerighteq \talloblong_{i \in I}(io_{i}\to Q_{i})\\
%  &|\langle F(\dot{s},s)=0\& B \rangle \trianglerighteq_{d}Q\\
%  S::=&P|S\| S
% \end{split}
%\end{equation*}
  \[
\begin{array}{l}
  P ::= \pskip \mid x :=e  \mid \pwait \ d \mid ch?x \mid
         ch!e \mid P;Q  \mid B \rightarrow P\mid  P \sqcap Q \mid P^* \\
%         \qquad\\
      \qquad \mid \talloblong_{i\in I} io_i\rightarrow P_i\mid \evolution{F}{\s}{B} \mid
\exempt{\evolution{F}{\s}{B}}{i\in I}{io_i}{Q_i}\\
      %\qquad \mid \evolution{F}{s}{B}\trianglerighteq_{d}Q\\
  S ::= P \mid S\| S
  \end{array}
  \]
where $x,\s$ for variables and vectors of variables, respectively, $B$ and $e$ are boolean and arithmetic expressions, $d$ is a non-negative real constant, $ch$ is the channel name, $io_{i}$ stands for a communication event, i.e., either $ch_{i}?x$ or $ch_{i}!e$, $P,Q,Q_{i}$ are sequential process terms, and $S$ stands for an HCSP process term.
Given an HCSP process $S$, we define $\VAR(S)$ for the set of variables in $S$,  and $\Sigma(S)$ the set of channels occurring in $S$, respectively.
% and, $ch_{i}\ast$ stands for a communication event, i.e., either $ch_{i}?x$ or $ch_{i}!e$.
The informal meanings of the individual constructors are as follows:
\begin{itemize}
\item $\pskip$, $x := e$, $\pwait \ d$, $ch?x$, $ch!e$, $P;Q$,  $P \sqcap Q$, and $\talloblong_{i\in I} io_i\rightarrow P_i$ are defined as usual. $B \rightarrow P$ behaves as $P$ if $B$ is true, otherwise terminates.

% terminates immediately having no effect on variables.
%\item $x := e$ assigns the value of expression $e$ to $x$ and then terminates.
%\item $\pwait \ d$ will keep idle for $d$ time units keeping variables unchanged.
%\item $ch?x$ receives a value along channel $ch$ and assigns it to $x$.
%\item $ch!e$ sends the value of $e$ along channel $ch$. A communication takes place when both the sending and the receiving parties are ready, and may cause one side to wait.
%\item The sequential composition $P;Q$ behaves as $P$ first, and if it terminates, as $Q$ afterwards.
%\item The alternative $B \rightarrow P$ behaves as $P$ if $B$ is true; otherwise it terminates immediately.
%\item $P \sqcup Q$ denotes internal choice. It behaves as either $P$ or $Q$, and the choice is made by the process. \item    $\talloblong_{i\in I} io_i\rightarrow P_i$ denotes external choice. If a communication $io_j$ for some $j$ occurs earlier than the others, then the corresponding branch is chosen to execute. If there are multiple communications in $\{io_i\}_{i\in I}$ getting ready simultaneously earlier than others, then  one among them will be chosen randomly to execute.
\item For repetition $P^*$, $P$ executes for an arbitrary finite number of times. We assume an oracle $\repnum$, s.t.  for a given $P^*$ in the context process $S$, $\repnum(P^*, S)$ returns the upper bound of the number of times that $P$ is repeated in the context.
    %that for a given $P^*$, there always exists an upper bound $N$ s.t.  the number of times that $P$ is repeated will not exceed  $N$.
\item $\evolution{F}{\s}{B}$ is the continuous evolution statement. It forces the vector $\s$ of real variables to obey the differential equations $F$ as long as  $B$, which defines the domain of $\s$, holds, and terminates when $B$ turns false. Without loss of generality, we assume that the set of $B$ is open, thus the escaping point will be at the boundary of $B$.  The communication interrupt $\exempt{\evolution{F}{\s}{B}}{i\in I}{io_i}{Q_i}$ behaves like $\evolution{F}{\s}{B}$, except that the continuous evolution is preempted as soon as one of the communications $io_{i}$ takes place, which is followed by the respective $Q_{i}$. These two statements are the main extension of HCSP for describing continuous behavior.
    %Notice that, if the continuous part terminates before a communication from among $\{io_{i}*\}_{I}$ occurs, then the process terminates without communicating.
%\item $\evolution{F}{s}{B}\trianglerighteq_{d}Q$ behaves like $\evolution{F}{s}{B}$, if that continuous terminates before $d$ time units. Otherwise, after $d$ time units of evolution according to $F$, it moves on to execute $Q$.
\item $S_1\| S_2$ behaves as if $S_{1}$ and $S_{2}$ run independently except that all communications along the common channels connecting $S_{1}$ and $S_{2}$ are to be synchronized.  $S_{1}$ and $S_{2}$ in parallel can neither share variables, nor input or output channels.
\end{itemize}

For better understanding of the HCSP syntax, we  model the water tank system \cite{Ahmad14}, for which two components \emph{Watertank} and \emph{Controller}, are composed in parallel. The HCSP model of the system is given by \emph{WTS} as follows:
\[\small
\begin{array}{l}
\textit{WTS} \ \ \ \ \ \ \ \ \ \Define \textit{Watertank} \| \textit{Controller} \\
\textit{Watertank} \Define v:=v_0; d:=d_0;\\
\ \ \ \ \ \ \ \ \ \ \ \ \ \ \ \ \ \ \ \ \ \ \ (v=1 \to \langle \dot{d}=Q_{max} - \pi r^2  \sqrt{2  g d}\rangle \trianglerighteq (wl!d \to cv?v); \\
\ \ \ \ \ \ \ \ \ \ \ \ \ \ \ \ \ \ \ \ \ \ \ \ v=0 \to \langle \dot{d}=- \pi  r^2  \sqrt{2  g  d}\rangle \trianglerighteq (wl!d \to cv?v))^* \\
\textit{Controller} \Define y:=v_0;x:=d_0;(\pwait \ p; wl?x; x\ge ub \to y:=0;\\
\ \ \  \ \ \ \ \ \ \ \ \ \ \ \ \ \ \ \ \ \  \ \ \ x \le lb \to y:=1; cv!y)^*
\end{array}
\]
where $Q_{max}$, $\pi$, $r$ and $g$ are system parameters, $v$ is the control variable which takes $1$ or $0$, depending  on whether the valve is open or not, $d$ is the water level of the $\textit{Watertank}$ and its dynamics depends on the value of $v$. $v_0$ and $d_0$ are the initial values of controller variable and water level, respectively.  Two channels, $wl$ and $cv$, are used to transfer the water level ($d$ in $\textit{Watertank}$) and control variable ($y$ in $\textit{Controller}$) between $\textit{Watertank}$ and $\textit{Controller}$, respectively. The control value is computed by the \emph{Controller} with a period of $p$. When the water level is less than or equal to  $lb$,  the control value is assigned to $1$, and when the water level is greater than or equal to $ ub$, the control value is assigned to $0$, otherwise, it keeps unchanged. Basically, based on the current value of $v$, $\textit{Watertank}$ and $\textit{Controller}$ run  independently for $p$ time, then $\textit{Watertank}$ sends the current water level to $\textit{Controller}$, according to which a new value of the control variable is generated and sent back to $\textit{Watertank}$, after that, a new period repeats.

\subsection{Transition system of HCSP}
 Given an HCSP process $S$, we can derive a transition system $T(S)=\langle Q,L,\to,Q^{0},Y,H\rangle $ from $S$ by the following procedure:
\begin{itemize}
\item the set of states $Q=(\subp(S)\cup\{\epsilon\})\times V(S)$, where $\subp(S)$ is the set of sub-processes of $S$, e.g., $\subp(S)=\{S, \pwait \ d,B \rightarrow P\} \cup \subp(P)$ for $S::=\pwait \ d;B \rightarrow P$, $\epsilon$ is introduced to represent the terminal process, meaning that the process has terminated, and $V(S) = \{v| v\in \VAR(S) \rightarrow \VAL\}$ is the set of evaluations of the variables in $S$, with $\VAL$ representing the value space of variables. Without confusion in the context, we often call an evaluation $v$ a (process) state. Given a state $q\in Q$, we will use $\first(q)$ and $\second(q)$ to return the first and second component of $q$, respectively.
     %In the mapped transition system, we denote a global clock $now$ as a special system variable which takes non-negative values and can only advanced to record the time in the execution of a process. Its value in a state $q$ can be got by $q(now)$.
\item The label set $L$ corresponds to the actions of HCSP, defined as  $L= \mathbb{R}^+_0 \cup \Sigma(S)\centerdot\{?,!\}\centerdot \RR \cup \{\tau\}$, where $d\in \mathbb{R}^+_0$ stands for the time progress, $ch?c, ch!c \in \Sigma(S)\centerdot\{?,!\}\centerdot \RR$ means that an input along channel $ch$ with value $c$ being received, an output along $ch$ with value $c$ being sent, respectively. Besides, the silent action $\tau$ represents a discrete non-communication action of HCSP, such as assignment, evaluation of boolean expressions, and so on.
\item $Q^{0}=\{(S, v) | v \in V(S)\}$,  representing that $S$ has not started to execute, and $v $ is the initial process state of $S$.
\item $Y=\overline{\VAL}$, represents the set of value vectors corresponding to $\VAR(S)$.
\item Given $q \in Q$, $H(q)=\vectorr(\second(q))$,  where function $\vectorr$ returns the value vector corresponding to the process state of $q$.
%as $Var(P)$ is the set of variables in $P$, $\qq(Var(P))$ denotes the values of variables in state $\qq$,
\item $\to$ is the transition relation of $S$, which is given next.
    %as different statements of HCSP may denote different transition relation in $T(P)$, we define a series of rules for generating transition relations corresponding to every kind of statements in Table~\ref{table:transitionrelation}:
 \end{itemize}

\subsubsection*{Sequential processes}

A transition relation of a sequential  HCSP process takes the form $(P,v)\xrightarrow{l}(P^{\prime},v^{\prime})$, indicating that starting from  state $v$, $P$ executes to $P'$ by performing action $l$, with the resulting state $v'$. Here we present the transition rules for continuous evolution as an illustration. Readers are referred to~\cite{ZWZ13} for the full details of the transition semantics, for both sequential and parallel HCSP processes.
\[
\begin{array}{c}
  \fracN{
\begin{array}{c}
   \forall d>0. \exists S(.):[0, d] \rightarrow \RR^n. (S(0)=v(\s)\land (\forall p\in[0,d).(F(\dot{S}(p),S(p))=0\\
   \land v[\s\mapsto S(p)](B)=true)))
\end{array}
 } {(\evolution{F}{\s}{B},v)\xrightarrow{d}(\evolution{F}{\s}{B},v[\s\mapsto S(d)])} \\[1.5em]
\fracN{v(B)=false} {(\evolution{F}{\s}{B},v)\xrightarrow{\tau}(\epsilon,v)}
\end{array}\]
%
%The  skip  and  assignment statements  are treated as $\tau$ actions and can be defined as usual. $v[x \mapsto e]$ represents a new evaluation that is same to $v$ except that $x$ is changed to the value of $e$.  The $\pwait\ d$ statement could be divided into any smaller intervals, i.e., any $d^{\prime}<d$, and when $d$ time passes, it terminates. For communication events $ch?x$ and $ch!e$,  two  rules are defined, one for waiting, in which $d$ could be any positive real number, and the other for synchronization. Moreover,  the action $ch?c$ stands for receiving a value $c$ from $ch$ and as a consequence of the communication, $x$ is assigned to $c$; $ch!v(e)$ stands for sending the value of $e$ under $v$, represented by $v(e)$, to channel $ch$.
%

For $\evolution{F}{\s}{B}$,
 for any $d\geq 0$,
it evolves for $d$ time units according to $F$
if $B$ evaluates to true within this period (the right end exclusive). In the rule, $S(\cdot): [0, d] \rightarrow \RR^n$ defines the trajectory  of the ODE $F$ with initial value $v(\s)$. Otherwise, by performing a $\tau$ action, the continuous evolution
terminates if $B$ evaluates to false.
%For communication interrupt, the process may evolve for $d$ time units if the continuous evolution can progress for $d$ time units, or it performs a communication $io_i$ for some $i \in I$ and the corresponding $Q_i$ is followed, or it performs a $\tau$ action and terminates when $B$ turns false. The transition rules for compound constructs can be defined as usual.

% Otherwise, the continuous evolution
%terminates at a point if $B$ evaluates to false at the point, or if $B$ evaluates to false
% at a positive open interval right to the point.two cases are considered: the first one is $P$ could run for a $d$ length interval with $v(B)=true$ at every instance, and the second is $v(B)=false$ at the start time, which generates a transition labeled with $\tau$. The two continuous statements with interrupts, $\evolution{F}{s}{B}\trianglerighteq_{d}Q$ and $\exempt{\evolution{F}{s}{B}}{i\in I}{io_i}{Q_i}$, are similar with $\evolution{F}{s}{B}$ except for taking the time and communication interrupts into consideration. For both of them, three kinds of rules needed. $P_{1};P_{2}$ is the sequential composition of $P_{1}$ and $P_{2}$, and according to the post state of $P_{1}$, two rules are used for transition generation as shown in the above formulas. $P^*$ can be interpreted as finite times of sequential composition of $P$.

\subsubsection*{Parallel composition}
Given two sequential processes $P_{1}$, $P_{2}$ and their transition systems $T(P_{1})=\langle Q_{1},L_{1},\to_{1},Q^{0}_{1},Y_{1},H_{1}\rangle $ and $T(P_{2})=\langle Q_{2},L_{2},\to_{2},Q^{0}_{2},Y_{2},H_{2}\rangle $, we can define the transition system of $P_{1}\|P_{2}$ as $T(P_{1}\|P_{2})=\langle Q,L,\to,Q,Y,H\rangle $, where:
\begin{itemize}
\item $Q=((\subp(P_1)\cup \{\epsilon\}) \| (\subp(P_2)\cup\{\epsilon\}))\times \{v_1 \uplus v_2 | v_1 \in V(P_1), v_2 \in V(P_2)\}$, where given two sets of processes $PS_1$ and $PS_2$, $PS_1 \| PS_2$ is defined as $\{\alpha\| \beta | \alpha \in PS_1 \wedge \beta \in PS_2\}$;  $v_1 \uplus v_2$ represents the disjoint union, i.e. $v_1 \uplus v_2 (x)$ is $v_1(x)$ if $x \in \VAR(P_1)$, otherwise $v_2(x)$.
\item $L=L_{1}\cup L_{2}$.
\item $Q^{0}=\{(P_{1}\|P_{2},  v^{0}_{1}\uplus v^{0}_{2}) | (P_i,  v^{0}_{i}) \in Q^{0}_{i} \mbox{ for $i=1, 2$}\}$.
\item $Y=Y_{1}\times Y_{2}$, the observation space of the parallel composition is obviously the Cartesian product of $Y_1$ and $Y_2$.
\item $H(q)= H_1(q) \times H_2(q)$,
 the observation function is the Cartesian product of the two component observation functions correspondingly.
\item $\to$ is defined based on the parallel composition of transitions of $L_1$ and $L_2$.
%Its definition is presented in Table~\ref{table:transitionrelationforparallel}.
%   Suppose two transitions $(P_{1},u)\xrightarrow{\alpha}(P^{\prime}_{1},u^{\prime})$ and $(P_{2},v)\xrightarrow{\beta}(P^{\prime}_{2},v^{\prime})$ occur for $P_1$ and $P_2$, respectively.
%    %The transition relation in $T(P_{1}\|P_{2})$ can be generated using the following transition rules (dual rules are omitted) in Table~\ref{table:transitionrelationforparallel}:
\end{itemize}

Suppose two transitions $(P_{1},u)\xrightarrow{\alpha}(P^{\prime}_{1},u^{\prime})$ and $(P_{2},v)\xrightarrow{\beta}(P^{\prime}_{2},v^{\prime})$ occur for $P_1$ and $P_2$, respectively. The rule for synchronization is given below:
\[   \fracN{\alpha=ch_{i}?c \land \beta=ch_{i}!e \land c=e}{(P_{1}\|P_{2},u\uplus v)\xrightarrow{\tau}(P^{\prime}_{1}\|P^{\prime}_{2},u^{\prime}\uplus v^{\prime})}\]

\subsection{Approximate bisimulation between HCSP processes}
%Based on the transition semantics of HCSP, the compressed  transition system for a process can be constructed.
  Let $P_1$ and $P_2$ be two HCSP processes, and  $h, \varepsilon$  the time and value precisions. Let $v_0$ be an arbitrary initial state. $P_1$ and $P_2$ are $(h, \varepsilon)$-\emph{approximately bisimilar}, denoted by $P_1\cong_{h, \varepsilon}P_2$, if $T(P_1)\cong_{h, \varepsilon}T(P_2)$, in which $T(P_1)$ and $T(P_2)$ are the $\tau$-compressed transition systems of $P_1$ and $P_2$ with the same initial state $v_0$, respectively.

{\small
  \begin{algorithm}[htb]
\caption{ Deciding approximately bisimilar between two HCSP processes}
\begin{algorithmic}[1]
\REQUIRE ~~~~
   Processes $P_1, P_2$, the initial state $v_0$, the time step $d$, and precisions $h$ and $\varepsilon$;
%\LASTCON ~~
%  if $P_{1}\cong_{h, \varepsilon}P_{2}$, return true, otherwise, return false;
\ENSURE ~~\\
 $T(P_m).Q^0 = \{(P_m, v_0)\}, T(P_m).T^0 = \emptyset$ for $m=1, 2$; $i=0$;
 \REPEAT
\STATE $T(P_m).T^{i+1} =  T(P_m).T^{i} \cup \{\compress{q}{l}{q'}  | \forall q \in T(P_m).Q^{i},  \mbox{  if } (\exists
l \in \{d, \tau\} \cup \Sigma(P_m)\centerdot\{?,!\}\centerdot \RR. \compress{q}{l}{q'}) \mbox{ or }
 (\exists
l=d'. l < d \wedge \compress{q}{l}{q'} \wedge \mbox{ not } (\compress{q}{d''}{}) \mbox{ for any $d''$ in $(d', d]$}) \mbox{ and } \second(q')(t^m_j) < T^m_j\}$; \\
\STATE $T(P_m).Q^{i+1} = T(P_m).Q^{i} \cup \states(T(P_m).T^{i+1})$;\\
\STATE $i\leftarrow i+1$;
\UNTIL { $T(P_m).T^{i}=T(P_m).T^{i-1}$ }\\
\STATE $T(P_m).Q = T(P_m).Q^i; T(P_m).T = T(P_m).T^i$;\\
\STATE
$\mathcal{B}^0_{h, \varepsilon}=\{(q_1,q_2)\in T(P_1).Q\times T(P_2).Q | \textbf{d}(H_{1}(q_{1}),H_{2}(q_{2}))\le \varepsilon\}$; $i=0$;
\REPEAT
\STATE $\mathcal{B}^{i+1}_{h, \varepsilon}\leftarrow \{(q_1,q_2)\in B^{i}_{h, \varepsilon} | \compress{\forall q_{1}}{l}{_1 q^{\prime}_{1}} \in T(P_1).T$, $\exists q_{2}\xLongrightarrow{l'}_{2} q^{\prime}_{2} \in T(P_2).T$ s.t.  $(q^{\prime}_{1},q^{\prime}_{2})\in \mathcal{B}^i_{h, \varepsilon}$ and $\labeld(l, l') \le h$, and $\compress{\forall q_{2}}{l}{_2 q^{\prime}_{2}} \in T(P_2).T$, $\exists q_{1}\xLongrightarrow{l'}_{1} q^{\prime}_{1} \in T(P_1).T$ s.t.  $(q^{\prime}_{1},q^{\prime}_{2})\in \mathcal{B}^i_{h, \varepsilon}$ and $\labeld(l, l') \le h\}$;\\ \STATE $i\leftarrow i+1$;
\UNTIL { $\mathcal{B}^{i}_{h, \varepsilon}=\mathcal{B}^{i-1}_{h, \varepsilon}$ }
\STATE $\mathcal{B}_{h, \varepsilon}=\mathcal{B}^{i}_{h, \varepsilon}$;
\IF  {$ ((P_1, v_0),(P_2, v_0))\in \mathcal{B}_{h, \varepsilon} $}
\STATE return \textbf{true};
\ELSE
\STATE return \textbf{false};
\ENDIF
\end{algorithmic}\label{alg:compt-bisi}
\end{algorithm}
}

In Algorithm~\ref{alg:compt-bisi}, we consider the  $(h, \vare)$-approximate  bisimilation between  $P_1$ and $P_2$ for which all the ODEs occurring in $P_1$ and $P_2$ are \ggas. Suppose the set of ODEs occurring in $P_i$ is $\{F^i_1, \cdots, F^i_{ki}\}$, and the equilibrium points for them are $x^i_1, \cdots, x^i_{ki}$ for $i=1, 2$ respectively. As a result, for each ODE, there must exist a sufficiently large time, called \emph{equilibrium time}, s.t.  after the time, the distance between the trajectory and the equilibrium point is less than $\vare$. We denote the  equilibrium time for each $F^i_{j}$ for $j=1, \cdots, ki$ by $T^i_j$, respectively. Furthermore, in order to record the execution time of ODEs, for each ODE $F^i_{j}$, we introduce an auxiliary time variable $t^i_{j}$ and add $t^i_{j}:=0; \dot{t^i_{j}}=1$ to $F^i_{j}$ correspondingly.

Algorithm~\ref{alg:compt-bisi} decides whether $P_1$ and $P_2$ are $(h, \vare)$-approximately bisimilar. When $P_{1}\cong_{h, \varepsilon}P_{2}$, it returns \textbf{true}, otherwise, it returns \textbf{false}. Let $d$ be the discretized time step.
The algorithm is then taken in two steps. The first step (lines 1-6) constructs the transition systems for $P_1$ and $P_2$ with time step $d$. For $m=1, 2$, $T(P_m).Q$ and $T(P_m).T$ represent the reachable set of states and transitions of $P_m$, respectively, which are initialized as empty sets and then constructed iteratively. At each step $i$, a new transition can be a $d$ time progress,  a $\tau$ event, or a communication event. Besides, a transition can be a time progress less than $d$, which might be caused by the occurrence of a boundary interrupt or a communication interrupt during a continuous evolution. The new transition will be added only when the running time for each ODE $F^m_j$, denoted by $t^m_j$, is less than the corresponding equilibrium time. Therefore, for either process $P_m$, whenever some ODE runs beyond its equilibrium time, the set of reachable transitions reaches a fixpoint by allowing precision $\vare$ and will not be extended any more. The set of reachable states can be obtained by collecting the post states of reachable transitions. Based on Def.~\ref{definition:appbisimulation}, the second step (lines 7-17) decides whether the transition systems  for $P_1$ and $P_2$ are  approximately bisimilar with the given precisions.

%  Based on the Definition \ref{definition:appbisimulation} and Definition \ref{definition:TSappbisimulation}, we give an algorithm ($\textbf{Algorithm 1}$) to decide whether two  transition systems are approximately bisimilar with the given precisions.
%As shown by Algorithm 1, if the number of states is finite for both $T_1$ and $T_2$,  a fixed point will be reached in a finite number of steps. Otherwise, if either $T_1$ or $T_2$ is infinite, the algorithm may not reach a fixed point in a finite number of steps.  If the fixed point exists, it is the maximal $(h, \varepsilon)$-approximate simulation relation between $T_1$ and $T_2$. The correctness of
%Algorithm~\ref{alg:compt-bisi}
%is guaranteed by Theorem \ref{theorem:algorithmcorrevtness}.

The first part  (lines 1-6) of the algorithm computes the transitions of processes. For each process $P_m$, its complexity  is $O(|T(P_m).T|)$, which is $O(\lceil \frac{T_m}{d}\rceil + N_m)$, where $T_m$ represents the execution time of $P_m$ till termination or reaching the equilibrium time of some  ODE, and $N_m$ the number of atomic statements of $P_m$. The second part (lines 7-17) checks for $P_1$ and $P_2$ each pair of the states whose distance  is within $\varepsilon$ by traversing the outgoing transitions, to see if they are truly approximate bisimilar, till the fixpoint $\mathcal{B}_{h, \varepsilon}$ is reached. We can compute the time complexity to be $O(Q_1^2Q_2^2T_1T_2)$, where $Q_m$ and $T_m$ represent $O(|T(P_m).Q|)$ and $O(|T(P_m).T|)$ for $m=1, 2$ respectively.

\begin{theorem}[Correctness] \label{theorem:algorithmcorrevtness}
%Algorithm~\ref{alg:compt-bisi}
Algorithm~\ref{alg:compt-bisi} terminates, and for any $v_0$, $P_1 \cong_{h, \varepsilon} P_2$ iff
   $((P_1, v_0), (P_2, v_0))\in \mathcal{B}_{h, \varepsilon}$.
   %Let $\{\mathcal{B}^{i}_{h, \varepsilon}\}_{i\in \mathbb{N}}$ be the decreasing sequence of sets defined by Algorithm~\ref{alg:app-bisimilar} and $\mathcal{B}^{max}_{h, \varepsilon}$ be the maximal $(h, \varepsilon)$-approximate simulation relation between $T_1$ and $T_2$. Then, the following facts hold:
%\[
%\begin{array}{l}
%\forall i\in \mathbb{N}. \, \mathcal{B}^{max}_{h, \varepsilon} \subseteq \mathcal{B}^{i}_{h, \varepsilon} \wedge
%\bigcap^{i=+\infty}_{i=0}\mathcal{B}^{i}_{h, \varepsilon}=\mathcal{B}^{max}_{h, \varepsilon}.
%\end{array}
%\]
\end{theorem}

%% file: C5-Discretization-of-HCSP.tex
\section{Discretization of HCSP}
\label{section:discretizationofhcsp}

In this section, we consider the discretization of HCSP processes, by which the continuous dynamics is represented by discrete approximation. Let $P$ be an HCSP process and  $(h,\varepsilon$) be the precisions, our goal is to construct a discrete process $D$ from $P$, s.t.  $P$ is $(h,\varepsilon)$-bisimilar with $D$, i.e., $\bisimilar{P}{D}$ holds.

\subsection{Discretization of Continuous Dynamics}

Since most differential equations do not have explicit solutions, the discretization of the dynamics is normally given by discrete approximation. Consider the ODE $\dot{\xx} = \ff(\xx)$ with the initial value $\widetilde{\xx}_0 \in \RR^n$, and assume $X(t, \widetilde{\xx}_0)$ is the trajectory of the initial value problem along the time interval $[t_0, \infty)$. In the following discretization, assume $h$ and $\euler$ represent the time step size and the precision of the discretization, respectively. Our strategy is as follows:
 \begin{itemize}
   \item First, from the fact that $\dot{\xx} = \ff(\xx)$ is GAS, there must exist a sufficiently large $T$ s.t.   $\|X(t, \widetilde{\xx}_0)-\bar{\xx}\| < \euler$ holds when $t > T$, where $\bar{\xx}$ is an equilibrium point. As a result, after time $T$, the value of $\xx$ can be approximated by the equilibrium point $\bar{\xx}$ and the distance between the actual value of $\xx$ and $\bar{\xx}$ is always within $\euler$.

   \item Then, for the bounded time interval $[t_0, T]$, we apply Euler method to discretize the continuous dynamics.
 \end{itemize}

There are a range of different discretization  methods for  ODEs~\cite{Stoer13} and the Euler method is an effective one among them.  According to the Euler method, the ODE $\dot{\xx} = \ff(\xx) $ is discretized as
\[(\xx:= \xx+h\ff(\xx); \pwait\ h)^N\]
A sequence of approximate solutions $\{\xx_i\}$ at time stamps $\{h_i\}$ for $i=1, 2, \cdots, N$ with $ N= \lceil \frac{T-t_0}{h}\rceil$ are obtained, satisfying (define $\xx_0 = \widetilde{\xx}_0$):
\[h_i = t_0 + i*h \quad \xx_i = \xx_{i-1} + h \ff(\xx_{i-1}).\]
 $\|X(h_i, \widetilde{\xx}_0) - \xx_i\|$ represents the discretization error at time $h_i$. To estimate the global error of the approximation,  by Theorem \textbf{3} in~\cite{Platzer12}, we can prove the following theorem:
\begin{theorem}[Global error with an initial error]
   Let $X(t, \widetilde{\xx}_0)$ be a solution on $[t_0, T]$ of the initial value problem $\dot{\xx} = \ff(\xx), \xx(t_0)=\widetilde{\xx}_0$, and $L$  the Lipschitz constant s.t.  for any compact set $S$ of $\RR^n$, $\|\ff(\yy_1) - \ff(\yy_2)\| \leq L\|\yy_1 - \yy_2\|$ for all $\yy_1, \yy_2 \in S$. Let $\xx_0 \in \RR^n$ satisfy $\|\xx_0 - \widetilde{\xx}_0\| \le \euler_1$. Then there exists an $h_0 > 0$, s.t.  for all $h$ satisfying $0<h\leq h_0$, and for all $n$ satisfying $nh \leq (T-t_0)$, the sequence $\xx_n = \xx_{n-1} + h\ff(\xx_{n-1})$ satisfies:
   \[\|X(nh, \widetilde{\xx}_0) - \xx_n\| \leq e^{(T-t_0)L}\euler_1 + \frac{h}{2}\max_{\zeta \in [t_0, T]}{\|X''(\zeta, \widetilde{\xx}_0)\|}\frac{e^{L(T-t_0)} -1}{L}\]
   \vspace*{-5mm}
\label{theorem:globalerror}
\end{theorem}
%It should be noticed that different from Theorem \textbf{3} of~\cite{Platzer12},  the initial discretization error is not 0, but within $\euler_1$.

By Theorem~\ref{theorem:globalerror} and the property of GAS, we can prove the following main theorem. % for the approximation of an ODE.
\begin{theorem}[Approximation of an ODE] Let $X(t, \widetilde{\xx}_0)$ be a solution on $[t_0, \infty]$ of the initial value problem $\dot{\xx} = \ff(\xx), \xx(t_0)=\widetilde{\xx}_0$, and $L$  the Lipschitz constant.   Assume $\dot{\xx} = \ff(\xx)$ is GAS with the equilibrium point $\bar{\xx}$.   Then for any precision $\euler>0$, there exist $h>0, T>0$ and $\euler_1 >0$ s.t.  $\dot{\xx} = \ff(\xx), \xx(t_0)=\widetilde{\xx}_0$ and  $\xx:=\xx_0; (\xx:= \xx+h\ff(\xx); \pwait\ h)^N;  \xx:=\bar{\xx}; \nstop$ with $N=\lceil \frac{T-t_0}{h}\rceil$ are $(h,\euler)$-approximately bisimilar, in which $\|\xx_0 - \widetilde{\xx}_0\| < \euler_1$ holds, i.e., there is an error between the initial values.
\label{theorem:approximate}
\end{theorem}

\subsection{Discretization of HCSP}

We continue to consider the discretization of HCSP processes, among which any arbitrary number of ODEs,  the discrete dynamics, and communications are involved. Below, given an HCSP process $P$, we use $\discrete{h}{\varepsilon}{P}$ to represent the discretized process of $P$,  with parameters $h$ and $\varepsilon$ to denote the step size and the precision (i.e.  the maximal ``distance'' between states in $P$ and $\discrete{h}{\varepsilon}{P}$), respectively.

Before giving the discretization of HCSP processes, we need to introduce the notion of readiness variables.
%First, we introduce a clock variable $t$ to represent the global time of the process. Then,
In order to express the readiness information of communication events, for each channel $ch$, we introduce two boolean variables $ch?$ and $ch!$, to represent whether the input and output events along $ch$ are ready to occur. We will see that in the discretization, the readiness information of partner events is necessary to specify the behavior of communication interrupt. %Finally, in order to approximate the conditions occurring in %processes, we define the approximation of boolean expressions.

Table~\ref{table:discretizationofHCSP} lists the definition of $\discrete{h}{\varepsilon}{P}$.
For each rule, the original process is listed above the line, while the discretized process is defined below the line.
For $\pskip, x:=e$ and $\pwait \ d$, they are kept unchanged in the discretization.  For input $ch?x$, it is discretized as itself, and furthermore, before $ch?x$ occurs,  $ch?$ is assigned to 1 to represent that $ch?x$ becomes ready, and in contrary,  after $ch?x$ occurs, $ch?$ is reset to 0. The output $ch!e$ is handled similarly. The compound constructs, $P; Q$, $P \sqcap Q$, $P^*$ and $P \| Q$ are discretized inductively according to their structure. For $B \rightarrow P$, $B$ is still approximated to $B$ and $P$ is discretized inductively. For external choice $\talloblong_{i\in I} io_i\rightarrow P_i$, the readiness variables $io_i$ for all $i\in I$ are set to 1 at first, and after the choice is taken, all of them are reset to 0 and the corresponding process is discretized. Notice that because $I$ is finite, the $\forall$ operator is defined as an abbreviation of the conjunction over $I$.

Given a boolean expression $B$ and a precision $\varepsilon$, we define $N(B, \varepsilon)$ to be a boolean expression which holds in the $\varepsilon$-neighbourhood of $B$. For instance, if $B$ is $x>2$, then  $N(B, \varepsilon)$ is $x>2-\varepsilon$.
For a continuous evolution  $\evolutionn{\dot{\xx}=\ff(\xx)}{B}$,  under the premise that $\dot{\xx}=\ff(\xx)$ is GAS, there must exists time $T$ such that when the time is larger than $T$, the distance between the actual state of $\xx$ and the equilibrium point, denoted by $\bar{\xx}$, is less than $\varepsilon$. Then according to  Theorem ~\ref{theorem:globalerror},
$\evolutionn{\dot{\xx}=\ff(\xx)}{B}$ is discretized as follows: First, it is a repetition of the  assignment to $\xx$ according to the Euler method for at most $\lceil \frac{T}{h}\rceil$ number of times, and then followed by the assignment of $\xx$  to the equilibrium point and stop forever. Both of them are guarded by the condition $N(B, \varepsilon)$.
%when $N(B, \varepsilon)$ holds, then if $t$ is less than or equal to $T$, then $\xx$ is assigned to $\xx+h \ff(\xx)$ by applying the Euler method, otherwise, it is assigned to the equilibrium point $\bar{\xx}$, and at the end,  $\pwait\ h$ is added to represent that this value will be kept unchanged for $h$ time units. Here  As seen from the discretization, $N(B, \varepsilon)$ is checked every $h$ time units, as a result,  there might be a delay of time, that is within $h$, in noticing that $N(B, \varepsilon)$ turns from 1 to 0.
For a communication interrupt $\exempt{\evolutionn{\dot{\xx}=\ff(\xx)}{B}}{i\in I}{io_i}{Q_i}$, suppose $T$ is sufficiently large s.t. when the time is larger than $T$, the distance between the actual state of $\xx$ and the equilibrium point, denoted by $\bar{\xx}$, is less than $\varepsilon$, and furthermore, if the interruption occurs, it must occur before $T$, and let $\overline{ch\ast}$ be the dual of $ch\ast$, e.g., if $ch\ast =ch?$, then $\overline{ch\ast}=ch!$ and vice versa. After all the readiness variables corresponding to $\{io_i\}_I$ are set to 1 at the beginning, the discretization is taken by the following steps:  first, if $N(B, \varepsilon)$ holds and no communication among $\{io_i\}_{i \in I}$ is ready,  it executes following the discretization of continuous evolution, for at most $\lceil \frac{T}{h}\rceil$ number of steps; then if $N(B, \varepsilon)$ turns false without any communication occurring, the whole process terminates and meanwhile the readiness variables are reset to 0; otherwise if some communications get ready, an external choice between these ready communications is taken, and then, the readiness variables are reset to 0 and the corresponding $Q_i$ is followed; finally, if the communications never occur and the continuous evolution never terminates, the continuous variable is assigned to the equilibrium point and the time progresses forever. It should be noticed that, the readiness variables of the partner processes will be used to decide whether a communication is able to occur. They are shared between parallel processes, but will always be written by one side.
\begin{table}[t]
\small
\centering
\begin{tabular}{c}
\hline \\
$\fracN{\pskip}{\pskip}\quad$
$\fracN{x:=e}{x:=e}\quad$
$ \fracN{\pwait \ d}{ \pwait \ d }$ \\[1.0em]
$ \fracN{ch?x}{ch?:=1; ch?x;  ch?:=0}\quad$
$ \fracN{ch!e}{ch!:=1; ch!e; ch!:=0}$ \\[1.0em]
$\fracN{P; Q}{\discrete{h}{\varepsilon}{P}; \discrete{h}{\varepsilon}{Q}} \quad \fracN{B \rightarrow P}{B \rightarrow \discrete{h}{\varepsilon}{P}}\quad$
$ \fracN{P \sqcap Q}{\discrete{h}{\varepsilon}{P} \sqcap \discrete{h}{\varepsilon}{Q}}$ \\[1.0em]
%$ \fracN{\evolution{F}{s}{B}}{(N(B,\varepsilon) \rightarrow (s:=TE(s_0,h,F(\dot{s},s)=0); \pwait \ h;s_0:=s))^{\ast}}$ \\[1.0em]
$\fracN{\talloblong_{i\in I} io_i\rightarrow P_i}{\forall i\in I. io_i :=1; \talloblong_{i\in I} io_i\rightarrow (\forall i\in I. io_i :=0; \discrete{h}{\varepsilon}{P_i})}$\\[1.0em]
$ \fracN{\evolutionn{\dot{\xx}=\ff(\xx)}{B}}{(N(B,\varepsilon) \rightarrow (\xx:=\xx+h \ff(\xx); \pwait \ h))^{\lceil\frac{T}{h}\rceil}; N(B, \varepsilon) \rightarrow (\xx:=\bar{\xx}; \nstop)}$ \\[1.8em]
% \frac{\evolution{F}{s}{B}\trianglerighteq_{d}Q}
%  {\begin{scriptsize}
%  \begin{split}
%  sumT:=0;(sumT\ge d\rightarrow Q;N(B,\varepsilon)\land sumT<d \rightarrow \\ (wait \ h;sumT:=sumT+h;s:=TE(s_0,h,F(\dot{s},s)=0);s_0:=s))^{\ast}
%   \end{split}
%   \end{scriptsize}
%   }\\ \\
%$ \fracN{\exempt{\evolution{F}{s}{B}}{i\in I}{io_i}{Q_i}}
%  {%\begin{scriptsize}
%  \begin{array}{c}
%  \forall i\in I,n(io_i):=1;(\exists i\in I.(io_i!\land io_i?)\rightarrow (io_i; \forall i\in I,n(io_i):=0;\discrete{h}{\varepsilon}{Q_i});N(B,\varepsilon)\\ \land \forall i\in I.((1-io_i!)\lor(1-io_i?)) \rightarrow (s:=TE(s_0,h,F(\dot{s},s)=0);wait \ h;s_0:=s))^{\ast}
%   \end{array}
%   %\end{scriptsize}
%   }$ \\[2.0em]
   %$ \fracN{\exempt{\evolutionn{\dot{\xx}=\ff(\xx)}{B}}{i\in I}{io_i}{Q_i}}
%  {%\begin{scriptsize}
%  \begin{array}{c}
%  \forall i\in I. io_i :=1; (\exists i. io_i \wedge \overline{io_i} \rightarrow (\sqcap_{i\in I} io_i \rightarrow \discrete{h}{\varepsilon}{Q_i}; \forall i\in I. io_i :=0);\\
%   \forall i\in I. io_i \wedge \neg \overline{io_i} \rightarrow (N(B, \varepsilon) \rightarrow
%   (t \leq T \rightarrow x:=\xx+h \ff(\xx); t> T \rightarrow x:=0; \pwait\ h)))^{\ast}
%   \end{array}
%   %\end{scriptsize}
%   }$ \\[2.0em]
$ \fracN{\exempt{\evolutionn{\dot{\xx}=\ff(\xx)}{B}}{i\in I}{io_i}{Q_i}}
  {%\begin{scriptsize}
  \begin{array}{c}
  \forall i\in I. io_i :=1; (N(B, \varepsilon) \rightarrow
   \forall i\in I. io_i \wedge \neg \overline{io_i} \rightarrow
    (\xx:=\xx+h \ff(\xx); \pwait\ h))^{\lceil\frac{T}{h}\rceil};\\
    \neg N(B, \varepsilon) \wedge \forall i\in I. io_i \wedge \neg \overline{io_i} \rightarrow  \forall i\in I. io_i :=0;\\
     \exists i. io_i \wedge \overline{io_i} \rightarrow (\talloblong_{i\in I} io_i \rightarrow (\forall i\in I. io_i :=0; \discrete{h}{\varepsilon}{Q_i})); \\
    (N(B, \varepsilon) \wedge \forall i\in I. io_i \wedge \neg \overline{io_i}) \rightarrow (\xx:=\bar{\xx}; \nstop);
   \end{array}
   %\end{scriptsize}
   }$ \\[3.0em]
% $\fracN{P;Q}{\discrete{h}{\varepsilon}{P};\discrete{h}{\varepsilon}{Q}}\quad$
$ \fracN{P^{\ast}}{(\discrete{h}{\varepsilon}{P})^{\ast}}\quad$
$\fracN{P\|Q}{\discrete{h}{\varepsilon}{P}\|\discrete{h}{\varepsilon}{Q}}$ \\[1.0em]
\hline\\
\end{tabular}
\caption{The rules for discretization of HCSP}
\label{table:discretizationofHCSP}

\end{table}

Consider the water tank system introduced in Sec. \ref{section:tsofhcsp}, by using the rules in Table~\ref{table:discretizationofHCSP}, a discretized system $\textit{WTS}_{h,\varepsilon}$ is obtained as follows:
\[\small
\begin{array}{l}
\textit{WTS}_{h,\varepsilon} \ \ \ \ \ \ \ \ \ \Define \textit{Watertank}_{h,\varepsilon} \| \textit{Controller}_{h,\varepsilon} \\
\textit{Watertank}_{h,\varepsilon} \Define v:=v_0; d:=d_0;( v=1 \to (wl!:=1; \\
\ \ \  \ \ \ \ \ \ \ \ \ \ \ \ \ \ \ \ \ \ \ \ \ \ \ \ (wl! \land \lnot wl? \to (d=d+h(Q_{max}-\pi r^2\sqrt{2gd});wait \ h;))^{\lceil \frac{T_1}{h}\rceil};\\
\ \ \  \ \ \ \ \ \ \ \ \ \ \ \ \ \ \ \ \ \ \ \ \ \ \ \ \ wl! \land wl? \to (wl!d;wl!:=0;cv?:=1;cv?v;cv?:=0);\\
\ \ \  \ \ \ \ \ \ \ \ \ \ \ \ \ \ \ \ \ \ \ \ \ \ \ \ \ wl! \land \lnot wl? \to(d=Q^2_{max}/2g\pi^2r^4;\nstop));\\
 \ \ \ \ \ \ \ \ \ \ \ \ \ \ \ \ \ \ \ \ \ \ \ \  \ \ \ v=0 \to (wl!:=1; \\
\ \ \  \ \ \ \ \ \ \ \ \ \ \ \ \ \ \ \ \ \ \ \ \  \ \ \ (wl! \land \lnot wl? \to (d=d+h(-\pi r^2\sqrt{2gd});wait \ h;))^{\lceil \frac{T_2}{h}\rceil};\\
\end{array}
\]
\[\small
\begin{array}{l}
\
 \ \ \ \ \ \ \ \ \  \ \ \ \ wl! \land wl? \to (wl!d;wl!:=0;cv?:=1;cv?v;cv?:=0);\\
\ \ \ \ \ \ \ \ \ \  \ \ \ \ wl! \land \lnot wl? \to(d=0;\nstop))
)^*
\end{array}
\]
\[
\begin{array}{l}
\textit{Controller}_{h,\varepsilon} \Define y:=v_0;x:=d_0;(wait \ p; wl?:=1;wl?x;wl?:=0; \\
\ \ \ \ \ \ \ \ \ \ \ \ \ \ \ \ \ \ \ \ \ \ \ \ \ \ \ \ x\ge ub \to y:=0;x \le lb \to y:=1;cv!:=1;cv!y;cv!:=0)^*
\end{array}
\]

%Firstly, we show how to construct an approximate process, $\discrete{h}{\varepsilon}{P}$ with parameters $h$ and $\varepsilon$, for an HCSP process $P$ as shown in Table~\ref{table:discretizationofHCSP}. In $\discrete{h}{\varepsilon}{P}$, $h$ is the time interval and $\varepsilon$ is the maximal ``distance'' between sates in $P$ and $\discrete{h}{\varepsilon}{P}$.

\subsection{Properties} Before giving the main theorem, we introduce some notations.
In order to keep the consistency between the behavior of an HCSP process and its discretized process, we introduce the notion of $(\delta, \epsilon)$-robustly safe. First, let $\phi$ denote a formula and $\epsilon$ a precision, define $N(\phi, -\epsilon)$ as the set $\{ \xx | \xx \in \phi \wedge \forall \yy \in \neg \phi. \|\xx-\yy\| > \epsilon\}$. Intuitively, when $\xx \in N(\phi, -\epsilon)$, then $\xx$ is inside $\phi$ and moreover the distance between it and the boundary of $\phi$ is greater than $\epsilon$.

\begin{definition}[$(\delta, \epsilon)$-robustly safe]
   An HCSP process $P$ is $(\delta, \epsilon)$-robustly safe, for a given initial state $v_0$, a time precision $\delta>0$ and a value precision $\epsilon>0$, if the following two conditions hold:
   \begin{itemize}
      \item for every continuous evolution $\evolutionn{\dot{\xx}=\ff(\xx)}{B}$ occurring in $P$, when  $P$ executes up to $\evolutionn{\dot{\xx}=\ff(\xx)}{B}$ at time $t$ with state $v$,  if $v(B) = false$, then there exists $\widehat{t} > t$ with $\widehat{t} -t < \delta$ s.t.  for any $\sigma$ satisfying $\dd(\sigma, v[\xx \mapsto X(\widehat{t}, \widetilde{\xx}_0)]) < \epsilon$, $\sigma\in N(\neg B, -\epsilon)$, where $X(t, \widetilde{\xx}_0)])$ is the solution of $\dot{\xx}=\ff(\xx)$ with initial value $\widetilde{\xx}_0  = v_0(\xx)$;
    \item for every alternative process $B \rightarrow P$ occurring in $S$, if $B$ depends on continuous variables of $P$, then when $P$ executes up to $B \rightarrow P$ at state $v$,  $v \in N(B, -\epsilon)$ or $v\in N(\neg B, -\epsilon)$.
   \end{itemize}

\end{definition}
As a result, when $P$ is discretized with a time error less than $\delta$ and a value error less than $\epsilon$, then $P$ and its discretized process have the same control flow.
The main theorem is given below.
\begin{theorem}
 Let $P$ be an HCSP process and $v_0$ is the initial state. Assume $P$ is $(\delta, \epsilon)$-robustly safe with respect to $v_0$. Let $0< \varepsilon < \epsilon$ be a precision. If for any  ODE $\dot{\xx}=\ff(\xx)$ occurring in $P$, $\ff$ is Lipschitz continuous and  $\dot{\xx}=\ff(\xx)$ is  GAS\ with $\ff(\bar{\xx})=0$ for some $\bar{\xx}$, then there exist $h >0$  and the equilibrium time for each ODE $F$ in $P$, $T_F>0$, s.t.   $P\cong_{h,\varepsilon}\discrete{h}{\varepsilon}{P}$.
 \label{theorem:app-process}
\end{theorem}

We can compute that, the relation
$\mathcal{L} \delta + Mh \leq \vare$  holds for some constants $\mathcal{L}$ and $M$. Especially, $\mathcal{L}$ is the maximum value of the first derivative of $\xx$ with respect to $t$.
More details can be found in ~\cite{reportpaper}.

%% file: C6-Case-study.tex
\section{Case study}
\label{section:casestudy}
In this section, we illustrate  our method through the safety verification of the water tank system, $\textit{WTS}$, that is introduced in Sec.~\ref{section:tsofhcsp}. The safety property is to maintain the value of $d$ within $[low,high]$, which needs to compute the reachable set of $\textit{WTS}$. However, it is usually difficult because of the complexity of the system. Fortunately, the reachable set of the  discretized $\textit{WTS}_{h,\varepsilon}$ in Sec. \ref{section:discretizationofhcsp} could be easily obtained. Therefore, we can verify the original system $\textit{WTS}$ through the discretized one, $\textit{WTS}_{h,\varepsilon}$, as follows.

\begin{table}%{r}{.5\linewidth}
%\begin{table}[!h]
\begin{center}
\vspace*{-5mm}
\begin{tabular}{c|c|c|c}
 \hline
 $\varepsilon$ & $h$  & $\textit{Reach}(\textit{WTS}_{h,\varepsilon})$ &$\textit{Reach}(\textit{WTS})$ \\
 \hline
 0.2 & 0.2 & [3.41, 6.5] &[3.21, 6.7] \\
 0.1 & 0.05 & [3.42, 6.47] &[3.32, 6.57] \\
 0.05 & 0.01 & [3.43, 6.46] &[3.38, 6.51] \\
 \hline
 \end{tabular}
\end{center}
\caption{The reachable set for different precisions}
\label{table:results}
%\end{table}
\vspace*{-5mm}
\end{table}
In order to analyze the system, first of all, we set the values of parameters to $Q_{max}=2.0$, $\pi=3.14$, $r=0.18$, $g=9.8$, $p=1$, $lb=4.1$, $ub=5.9$, $low = 3.3$, $high=6.6$, $v_0=1$, and $d_0=4.5$ (units are omitted here). Then, by simulation, we compute the values of $\delta$ and $\epsilon$ as $0.5$ and $0.24$, s.t.  $\textit{WTS}$ is $(\delta, \epsilon)$-robustly safe.
%Given an initial state $(v_0,d_0)$ and the values of all parameters in $WTS$, we can estimate a pair of $(\delta, \epsilon)$ that guarantees that $\textit{WTS}$ is $(\delta, \epsilon)$-robustly safe with respect to its initial state. In addition, it is easy to check that for every ODE in $WTS$, it is Lipschitz continuous and \ggas. Therefore, according to Theorem \ref{theorem:app-process}, for a given $\varepsilon$ with $0< \varepsilon < \epsilon$, there exists $h$ with $h>0$ and $T_i>0$ for $i\in{1,2}$, s.t.  $\textit{WTS}\cong_{h,\varepsilon}\textit{WTS}_{h,\varepsilon}$.
By Theorem \ref{theorem:app-process}, for a given $\varepsilon$ with $0< \varepsilon < \epsilon$, since $\dot{d}$ and $d$ are monotonic for both ODEs, we can compute a $h>0$  s.t.  $\textit{WTS}\cong_{h,\varepsilon}\textit{WTS}_{h,\varepsilon}$. For different values of $\varepsilon$ and $h$, $\textit{Reach}(\textit{WTS}_{h,\varepsilon})$ could be computed, and then based on Theorem \ref{theorem:reachsetrelation}, we can obtain $\textit{Reach}(\textit{WTS})$. Table \ref{table:results} shows the results for different choices of $\varepsilon$ and $h$. As seen from the results, when the values of precisions become smaller,   $\textit{Reach}(\textit{WTS}_{h,\varepsilon})$ and $\textit{Reach}(\textit{WTS})$ get closer and tighter. For the smaller precisions, i.e., $(\varepsilon=0.1,h=0.05)$ and $(\varepsilon=0.05,h=0.01)$, the safety property of the system is proved to be true. However, for $(\varepsilon=0.2,h=0.2)$, the safety property of the system can not be promised.

%% file: C7-Conclusion.tex
\section{Conclusion}
\label{section:conclusion}

Approximate bisimulation is a useful notion for analyzing complex dynamic systems via simpler abstract systems.
%When the  distance between the original system and the abstract system is sufficiently small, then the original system is approximately equal to the abstract system, and as a result, the results of analysis on the abstract system can be preserved in the original system.
In this paper, we define the approximate bisimulation of hybrid systems modelled  by HCSP, and present an algorithm for deciding whether two HCSP processes are approximately bisimilar. We have proved that if all the ODEs are \ggas, then the algorithm terminates in a finite number of steps. Furthermore, we define the discretization of HCSP processes, by representing the continuous dynamics by Euler approximation. We have proved  for an HCSP process that, if the process is robustly safe, and if each ODE occurring in the process is Lipschitz continuous and \ggas, then there must exist a discretization of the original HCSP process such that they are approximate bisimilar with the given precisions. Thus,   the results of analysis performed on the discrete system can be carried over into the original dynamic system, and vice versa.  At the end, we illustrate our method by presenting the discretization of a water tank example.
Note that \ggas \ and robust safety are very restrictive from a theoretical point of view, but most of real applications satisfy these conditions in practice.

Regarding future work, we will focus on the implementation, in particular, the transformation from HCSP to ANSI-C.
Moreover, it could be interesting to investigate  approximate bisimularity with time bounds so that the assumptions of GAS and robust safety can be dropped.
  In addition, it deserves to investigate richer refinement theories for HCSP based on the notion of approximately bisimulation,
  although itself can be seen as a refinement relation as discussed in process algebra.
  %in Future work include the investigation of the robustly safe  hybrid systems and the application of our method on more practical %systems.

%% file: Appendix.tex
\section*{Appendix}
\label{section:appendix}
\textbf{\emph{Proof of Lemma~\ref{lemma:union}}}: For any $(\qq_1,\qq_2)\in \bigcup_{i\in I}\mathcal{B}^i_{h, \varepsilon}$, there exists $i\in I$ such that $(\qq_1,\qq_2)\in \mathcal{B}^i_{h, \varepsilon}$. Then, $\textbf{d}(H_{1}(\qq_{1}),H_{2}(\qq_{2}))\le \varepsilon$. Moreover, for all $\compress{\qq_{1}}{l}{_1 \qq^{\prime}_{1}}$, there exists $\qq_{2}\xLongrightarrow{l^{\prime}}_{2} \qq^{\prime}_{2}$ such that $(\qq^{\prime}_1,\qq^{\prime}_2)\in \mathcal{B}^i_{h, \varepsilon} \subseteq \bigcup_{i\in I}\mathcal{B}^i_{h, \varepsilon}$ and $\labeld(l, l') \le h$, and for all $\compress{\qq_{2}}{l}{_2 \qq^{\prime}_{2}}$, there exists $\qq_{1}\xLongrightarrow{l^{\prime}}_{1} \qq^{\prime}_{1}$ such that $(\qq^{\prime}_1,\qq^{\prime}_2)\in \mathcal{B}^i_{h, \varepsilon} \subseteq \bigcup_{i\in I}\mathcal{B}^i_{h, \varepsilon}$ and $\labeld(l, l') \le h$. Therefore, $\bigcup_{i\in I}\mathcal{B}^i_{h, \varepsilon}$ is also a $(h,\varepsilon)$-approximate bisimulation relations between $T_1$ and $T_2$.\QEDB

\textbf{\emph{Proof of Theorem~\ref{theorem:reachsetrelation}}}: In order to prove Theorem \ref{theorem:reachsetrelation}, we need the following Lemma:
\begin{lemma}
If $T_{1}\cong_{h, \varepsilon}T_{2}$, then for all observation trajectory of $T_1$,
\[
  \begin{array}{l}
  \compress {\yy^0_1}{l_0}{\yy^1_1} \compress{}{l_1}{\yy^2_1} \compress{}{l_2}{...},
  \end{array}
\]
there exists an observation trajectory of $T_2$ with the sequence of labels
\[
  \begin{array}{l}
\yy^0_2\xLongrightarrow{l^{\prime}_0}\yy^1_2\xLongrightarrow{l^{\prime}_1}\yy^2_2\xLongrightarrow{l^{\prime}_2}...,
  \end{array}
\]
such that $\forall i\in \mathbb{N}$, $\textbf{d}(\yy^i_1,\yy^i_2)\le \varepsilon$ and $\labeld(l_i, l'_i) \le h$.
\label{theorem:trajectoryrelation}
\end{lemma}
\begin{proof}
For $\compress {\yy^0_1}{l_0}{\yy^1_1} \compress{}{l_1}{\yy^2_1} \compress{}{l_2}{...}$, there exists a state trajectory in $T_1$, $\compress {\qq^0_1}{l_0}{\qq^1_1} \compress{}{l_1}{\qq^2_1} \compress{}{l_2}{...}$, such that $\forall i\in \mathbb{N}$, $H_1(\qq^i_1)=\yy^i_1$. For $\qq^0_1\in Q^0_1$, then there exists $\qq^0_2\in Q^0_2$ such that $(\qq^0_1,\qq^0_2)$ is in the $\mathcal{B}_{h,\varepsilon}$. With the second property of Def.~\ref{definition:appbisimulation}, it can be shown by induction that there exists a state trajectory of $T_2$, $\qq^0_2\xLongrightarrow{l'_0}\qq^1_2\xLongrightarrow{l'_1}\qq^2_2\xLongrightarrow{l'_2}...$, such that $\forall i\in \mathbb{N}$, $(\qq^i_1,\qq^i_2)\in \mathcal{B}_{h,\varepsilon}$. Let $\yy^0_2\xLongrightarrow{l'_0}\yy^1_2\xLongrightarrow{l'_1}\yy^2_2\xLongrightarrow{l'_2}...$ be the associated observation trajectory of $T_2$ ($\forall i\in \mathbb{N}$, $H_2(\qq^i_2)=\yy^i_2$). Then,
\[
  \begin{array}{l}
\textbf{d}(\yy^i_1,\yy^i_2)=\textbf{d}(H_1(\qq^i_1),H_2(\qq^i_2))\le \varepsilon \ and \ \labeld(l_i, l'_i) \le h
  \end{array}
\]
for all $i\in \mathbb{N}$. \QEDB
\end{proof}

Therefore, from the definition of $Reach(T)$ and Lemma~\ref{theorem:trajectoryrelation}, it is straightforward that Theorem 1 holds.\QEDB

\textbf{\emph{Proof of Theorem~\ref{theorem:algorithmcorrevtness}}}: In order to ensure the termination of Algorithm~\ref{alg:compt-bisi}, both of the repeat pieces should be proved to be ended in finite steps. For the first loop (lines 1-5), $T(P_m).T^i$ is increasingly constructed, until a fixed point where $T(P_m).T^i=T(P_m).T^{i-1}$ reached. Since $T(P_m).T^i$ collect the feasible transitions in $T(P_m)$, which is the transition system generated from $P_m$, we just need to prove that $P_m$ terminates within a bounded time, with a given time step $d$. We assume all communication actions are feasible, i.e., they could happen in a limited time interval. So all HCSP processes without ODEs can terminate within a finite time duration. For processes with continuous evolution statements, as all ODEs in $P_m$ are GAS, we know that for each ODE, there exists an equilibrium time $T^i_j$ for $j=1,...,ki$, which is used to construct the set $T(P_m).T^{i+1}$. If the ODE is interrupted before $T^i_j$, the continuous evolution will terminate before $T^i_j$, which is a bounded time. Otherwise, if the ODE keeps evolution until its equilibrium time, according to the constrain defined in the process of $T(P_m).T^{i+1}$ construction, $\second(q')(t^i_j) < T^i_j$, no more transitions will be generated after $T^i_j$, which means the ODE terminates at $T^i_j$. Therefore, we can conclude that all HCSP processes $P_m$ can terminate in a bounded time, with a given time step $d$ and the GAS assumption. That is to say, the first repeat part (lines 1-5) can terminate in finite steps. Moreover,  since $T(P_m).Q$ is directly derived from $T(P_m).T$, $T(P_m).T$ and $T(P_m).Q$ are both finite sets, which are used for constructing another finite set, $\mathcal{B}^0_{h, \varepsilon}$ (line 7) that includes all compositional states that the distance between them is not greater than $\varepsilon$. In the second repeat section (lines 8-11), from the construct process of $\mathcal{B}^i_{h, \varepsilon}$ and the fact that $T(P_m).T$ and $T(P_m).Q$ have finite elements, it is clear that it can reach a fixed point in a finite number of steps. In conclusion, Algorithm~\ref{alg:compt-bisi} terminates.

In order to prove the second part of Theorem~\ref{theorem:algorithmcorrevtness}, we need to prove that $\mathcal{B}_{h, \varepsilon}=\bigcap^{i=N}_{i=0}\mathcal{B}^{i}_{h, \varepsilon}$, in which $N$ is the repeat time for the computation of $\mathcal{B}^{i}_{h, \varepsilon}$, is an approximate bisimulation relation, moreover, it is the maximal one. Assume that the maximal bisimulation relation with $(h, \varepsilon)$ is $\mathcal{B}^{max}_{h, \varepsilon}$, therefore, $\bigcap^{i=N}_{i=0}\mathcal{B}^{i}_{h, \varepsilon}=\mathcal{B}^{max}_{h, \varepsilon}$ need to be proved, i.e., $\bigcap^{i=N}_{i=0}\mathcal{B}^{i}_{h, \varepsilon} \subseteq \mathcal{B}^{max}_{h, \varepsilon}$ and $\mathcal{B}^{max}_{h, \varepsilon}  \subseteq \bigcap^{i=N}_{i=0}\mathcal{B}^{i}_{h, \varepsilon}$ should be hold simultaneously. According to the computation of $\mathcal{B}^{0}_{h, \varepsilon}$, $\mathcal{B}^{i+1}_{h, \varepsilon}$ and $\mathcal{B}_{h, \varepsilon}$, it is clear that $\mathcal{B}^{max}_{h, \varepsilon}  \subseteq \bigcap^{i=N}_{i=0}\mathcal{B}^{i}_{h, \varepsilon}$. Hence, we just need to show $\bigcap^{i=N}_{i=0}\mathcal{B}^{i}_{h, \varepsilon} \subseteq \mathcal{B}^{max}_{h, \varepsilon}$, i.e., $\bigcap^{i=N}_{i=0}\mathcal{B}^{i}_{h, \varepsilon}$ is a $(h, \varepsilon)$-approximate simulation relation between $T_1$ and $T_2$. For any $(q_1,q_2)\in \bigcap^{i=N}_{i=0}\mathcal{B}^{i}_{h, \varepsilon}$, then particularly $(q_1,q_2)\in \mathcal{B}^{0}_{h, \varepsilon}$. Hence, $\textbf{d}(H_{1}(q_{1}),H_{2}(q_{2}))\le \varepsilon$. As the sequence $\{\mathcal{B}^{i}_{h, \varepsilon}\}_{i\in [0,N]}$ is decreasing and approach a fixed point as $i$ increasing to $N$, so, for $N-1$, $\mathcal{B}^{N-1}_{h, \varepsilon}=\mathcal{B}^{N}_{h, \varepsilon}=\bigcap^{i=N}_{i=0}\mathcal{B}^{i}_{h, \varepsilon}$, Therefore, $\mathcal{B}^{N}_{h, \varepsilon}$ could be defined by
\[
\begin{array}{l}
  \mathcal{B}^{N}_{h, \varepsilon}=\{\textbf{d}(H_{1}(q_{1}),H_{2}(q_{2}))\le \varepsilon \ and \ \compress{\forall q_{1}}{l}{_1 q^{\prime}_{1}}, \ \exists q_{2}\xLongrightarrow{l'}_{2} q^{\prime}_{2}\ such \ that \\ \ \ \ \ \ \ \ \ \ \ \ \ \ (q^{\prime}_{1},q^{\prime}_{2}) \in \mathcal{B}^{N-1}_{h, \varepsilon}=\mathcal{B}^{N}_{h, \varepsilon}\  and \ \labeld(l, l') \le h,\ and \ \compress{\forall q_{2}}{l}{_2 q^{\prime}_{2}}, \ \exists q_{1}\xLongrightarrow{l'}_{1} q^{\prime}_{1}\ \\  \ \ \ \ \ \ \ \ \ \ \ \ \ such \ that \ (q^{\prime}_{1},q^{\prime}_{2})\in \mathcal{B}^{N-1}_{h, \varepsilon}=\mathcal{B}^{N}_{h, \varepsilon}\ and \ \labeld(l, l') \le h\}
\end{array}
\]
for all $(q_1,q_2)\in \mathcal{B}^{N}_{h, \varepsilon}$. It follows that $\bigcap^{i=N}_{i=0}\mathcal{B}^{i}_{h, \varepsilon}$ is a $(h, \varepsilon)$-approximate simulation relation between $T_1$ and $T_2$
\QEDB

\textbf{\emph{Proof of Theorem~\ref{theorem:globalerror}}}: Let $\{\widehat{\xx}_i\}$ and $h_0$ the approximate sequence and step size respectively in Theorem \textbf{3} of~\cite{Platzer12}, i.e., $\widehat{\xx}_0=\widetilde{\xx}_0$, and for all $h$ satisfying $0<h\leq h_0$, and for all $n$ satisfying $nh \leq (T-t_0)$, the sequence $\widehat{\xx}_n = \widehat{\xx}_{n-1} + hf(\widehat{\xx}_{n-1})$ satisfies:
   \[\|X(nh,\widetilde{\xx}_0) - \widehat{\xx}_n\| \leq \frac{h}{2}\max_{\zeta \in [t_0, T]}{\|X''(\zeta)\|}\frac{e^{L(T-t_0)} -1}{L}\]
As $\|\widehat{\xx}_0-\xx_0\|\le \euler_1$, and $\widehat{\xx}_1 = \widehat{\xx}_{0} + hf(\widehat{\xx}_{0})$, $\xx_1 = \xx_0 + hf(\xx_0)$, and $\ff$ is Lipschitz-continuous with Lipschitz-constant $L$, it is easy to conclude that $\|\widehat{\xx}_1-\xx_1\|\le (Lh+1)\euler_1$. Similarity, it can be concluded that $\|\widehat{\xx}_2-\xx_2\|\le (Lh+1)^2\euler_1$. By induction, we have:
\[\|\widehat{\xx}_n-\xx_n\|\le (Lh+1)^n\euler_1\]
Therfore, it is to see that:
\[
\begin{array}{l}
\|X(nh,\widetilde{\xx}_0)-\xx_n\|\le \|X(nh,\widetilde{\xx}_0) - \widehat{\xx}_n\|+ \|\widehat{\xx}_n-\xx_n\|\\
\ \ \ \ \ \ \ \ \ \ \ \ \ \ \ \ \ \ \ \le(Lh+1)^n\euler_1 + \frac{h}{2}\max\limits_{\zeta \in [t_0, T]}{\|X''(\zeta)\|}\frac{e^{L(T-t_0)} -1}{L}\\
\ \ \ \ \ \ \ \ \ \ \ \ \ \ \ \ \ \ \ \le (e)^{nhL}\euler_1 + \frac{h}{2}\max\limits_{\zeta \in [t_0, T]}{\|X''(\zeta)\|}\frac{e^{L(T-t_0)} -1}{L}\\
\ \ \ \ \ \ \ \ \ \ \ \ \ \ \ \ \ \ \ \le (e)^{(T-t_0)L}\euler_1 + \frac{h}{2}\max\limits_{\zeta \in [t_0, T]}{\|X''(\zeta)\|}\frac{e^{L(T-t_0)} -1}{L},
\end{array}\]
since $\|\xx+\yy\|\le\|\xx\|+\|\yy\|$ and $0 < 1+Lh \le e^{Lh}$ for $Lh >0$. \QEDB

\textbf{\emph{Proof of Theorem~\ref{theorem:approximate}}}: As $\dot{\xx} = f(\xx)$ is GAS with the equilibrium point $\bar{\xx}$, for a given $\euler >0$, we know that there exists $T>0$ when $t>T-t_0$ holds, $\|X(t, \widetilde{\xx}_0)-\bar{\xx}\|<\euler$, and $X(t, \widetilde{\xx}_0) \to \bar{\xx}$ when $t \to \infty$. Since $N=\lceil \frac{T-t_0}{h}\rceil$, we know that after the execution of $N$ numbers of Euler expansion with $h$ step length, the ODE reaches $X(T^{\prime}, \widetilde{\xx}_0)$ with $T^{\prime}=Nh\ge T-t_0$, which means the ``distance'' between the ODE and the equilibrium point $\bar{\xx}$ will no more greater than $\euler$ after $T^{\prime}$. The structure of the discretized process indicates that the transition system generated from it is a deterministic one. Also, since the Lipschitz continuous condition is assumed, the transition system of the ODE is deterministic. Therefore, it is clear that the ODE and the discretized process is $(h,\euler)$-approximate bisimilar on $[T^{\prime},\infty]$. Next, we prove they are $(h,\euler)$-approximate bisimilar on $[t_0,T^{\prime}]$.

From Def.~\ref{definition:TSappbisimulation}, we know that if there exists a $(h,\euler)$-approximate bisimulation relation, $\mathcal{B}_{h,\euler}$, between the transition systems of the ODE and the discretized process such that $(\widetilde{\xx}_0,\xx_0) \in \mathcal{B}_{h,\euler}$, the continuous and discretized ones are $(h,\euler)$-approximate bisimilar. Since $\|\widetilde{\xx}_0 - \xx_0\| < \euler_1$, we assume $\euler_1 \le \euler$ here, which satisfied the first condition in Def.~\ref{definition:appbisimulation}. In order to illustrate the existence of $\mathcal{B}_{h,\euler}$ and $(\widetilde{\xx}_0,\xx_0) \in \mathcal{B}_{h,\euler}$, we just need to ensure that the ``distance'' between the ODE and the discretized process never greater than $\euler$ within every interval $[ih,(i+1)h]$ for $i \in [0,N-1]$. The reason is that the transition systems of the ODE and the discretized process are both deterministic, and only time delay and assignment labels occur. It is easily understood: for given $h$ and $\euler_1$, we can compute the approximation Euler sequence $\{\xx_0,\xx_1,...,\xx_N\}$, if for any $\xx_{i+1}$ with $i \in [0,N-1]$, the ``distance'' between $X(t_i,\widetilde{\xx}_0)$ and $\xx_{i+1}$, in which $t_i\in [ih,(i+1)h]$, is not greater than $\euler$, i.e., $\|X(t_i,\widetilde{\xx}_0)-\xx_{i+1}\|\le \euler$ for every $t_i$ in $[ih,(i+1)h]$, we can see that for any two states satisfy $\|X(t_i,\widetilde{\xx}_0),\xx_{i+1}\| \le \euler$, $\compress{\forall X(t_i,\widetilde{\xx}_0)}{t}{_1 X(t_i+t,\widetilde{\xx}_0)}$ with $t \in \RR^+$, $\exists \xx_{i+1}\xLongrightarrow{\lceil \frac{t_i+t-(i+1)h}{h} \rceil h}_{2} \xx_k$ with $k=\lceil \frac{t_i+t}{h} \rceil h$ such that $\|X(t_i+t,\widetilde{\xx}_0)-\xx_k\|\le \euler$ and $\| \lceil \frac{t_i+t-(i+1)h}{h} \rceil h-t \|\le h$ hold, and $\compress{\forall \xx_{i+1}}{\tau}{_2 \xx_{i+2}}$, $\compress{\exists X(t_{i+1},\widetilde{\xx}_0)}{\tau^0}{_1 X(t_{i+1},\widetilde{\xx}_0)}$ such that $\|X(t_{i+1},\widetilde{\xx}_0)-\xx_{i+2}\|\le \euler$, and $\compress{\forall \xx_{i+1}}{h}{_2 \xx_{i+1}}$, $\compress{\exists X(t_{i},\widetilde{\xx}_0)}{(i+1)h-t_i}{_1 X((i+1)h,\widetilde{\xx}_0)}$ such that $\|X((i+1)h,\widetilde{\xx}_0)-\xx_{i+1}\|\le \euler$ and $\| (i+1)h-t_i-h\|=\|ih-t_i\|\le h$. By induction, the $(h,\euler)$-approximate bisimilar on $[t_0,T^{\prime}]$ is proved.

As mentioned above, with the assumption that $\euler_1 \le \euler$ and the ``distance'' limitation, we can indicate that the continuous and the discretized process are $(h,\euler)$-approximate bisimilar. For any $\euler > 0$, it always can choose a $\euler_1$ which makes $0 < \euler_1 \le \euler$ holds. In the following, we will illustrate the existence of $h$ that satisfies the ``distance'' assumption, i.e., $\|X(t_i,\widetilde{\xx}_0)-\xx_{i+1}\|\le \euler$ for $t_i\in [ih,(i+1)h]$ and $i \in [0,N-1]$.

First of all, we have
\[
\|X(t_i,\widetilde{\xx}_0)-\xx_{i+1}\| \le \|X(t_i,\widetilde{\xx}_0)-X((i+1)h,\widetilde{\xx}_0)\| + \|X((i+1)h,\widetilde{\xx}_0)-\xx_{i+1}\|
\]
From Theorem~\ref{theorem:globalerror}, the following inequality holds.
\[
 \|X((i+1)h,\widetilde{\xx}_0)-\xx_{i+1}\| \le e^{(T^{\prime}-t_0)L}\euler_1 + \frac{h}{2}\max_{\zeta \in [t_0, T^{\prime}]}{\|X''(\zeta)\|}\frac{e^{L(T^{\prime}-t_0)} -1}{L}
\]
Accordingly, $\|X(t_i,\widetilde{\xx}_0)-X((i+1)h,\widetilde{\xx}_0)\|$ denotes the ``distance'' between $X((i+1)h,\widetilde{\xx}_0)$ and $X(r_ih,\widetilde{\xx}_0)$, with real number $r_i \in [i,i+1]$. From Theorem (\textbf{7.1.1}) of~\cite{Stoer13}, we know that $X(t,\widetilde{\xx}_0)$ is continuous and continuously differentiable on $[t_0,T^{\prime}]$, hence on every segment $[ih,(i+1)h]$ for $0\le i \le (N-1)$. From the Lagrange Mean Value Theorem, we have $\|X(r_ih,\widetilde{\xx}_0)-X((i+1)h,\widetilde{\xx}_0)\|=\Delta h \|\ff(\zeta_i)\|$, where $\Delta h=(i+1-r_i)h$ and $\zeta_i$ is a point between $r_ih$ and$(i+1)h$. For all $r_i \in [i,i+1]$, there must exist a $r^{max}_i \in [i,i+1]$ such that $\|X(r^{max}_ih,\widetilde{\xx}_0)-X((i+1)h,\widetilde{\xx}_0)\|=D_ih$, where $D_i$ is a constant, is the maximal value on the segment $[ih,(i+1)h]$. In other words, $\|X(r_ih,\widetilde{\xx}_0)-X((i+1)h,\widetilde{\xx}_0)\| \le D_ih$ for any $r_i \in [i,i+1]$.

So, the ``distance'' between $X(t_i,\widetilde{\xx}_0)$ and $\xx_{i+1}$ on $[ih,(i+1)h]$ can be bounded using
\[
\|X(t_i,\widetilde{\xx}_0)-\xx_{i+1}\| \le D_ih + e^{(T^{\prime}-t_0)L}\euler_1 + \frac{h}{2}\max_{\zeta \in [t_0, T^{\prime}]}{\|X''(\zeta)\|}\frac{e^{L(T^{\prime}-t_0)} -1}{L}
\]
Then, on the whole interval $[t_0,T^{\prime}]$, the ``distance'' is bounded by
\[
\|X(t_i,\widetilde{\xx}_0)-\xx_{i+1}\| \le Mh + e^{(T^{\prime}-t_0)L}\euler_1 + \frac{h}{2}\max_{\zeta \in [t_0, T^{\prime}]}{\|X''(\zeta)\|}\frac{e^{L(T^{\prime}-t_0)} -1}{L}
\]
where $M=max(D_0,D_1,...,D_{N-1})$.

In conclusion, if
\[
 Mh + e^{(T^{\prime}-t_0)L}\euler_1 + \frac{h}{2}\max_{\zeta \in [t_0, T^{\prime}]}{\|X''(\zeta)\|}\frac{e^{L(T^{\prime}-t_0)} -1}{L} \le \euler
\]
holds, we can say that the original ODE and the discretized process are $(h,\euler)$-approximate bisimilar on $[t_0,T^{\prime}]$. Since we can always choose small enough $h$ and $\euler_1$ to make the inequality satisfied, the theorem is true.
\QEDB

\textbf{\emph{Proof of Theorem~\ref{theorem:app-process}}}: First of all, with proper equilibrium time $T_F$ for each ODE $F$ of $P$, given a step size $h$, we prove that the global discretized error of $P$ is $Mh$ for some constant $M$. As a result, when $h$ is sufficiently small (e.g., $h <\frac{\vare}{M}$), $Mh < \vare$ is guaranteed.
Now assume $P$ and $\discrete{h}{\varepsilon}{P}$ start to execute from the same initial state $\sigma$. Suppose $P$ executes to $P_1$ with state $\sigma_1$, and in correspondence, $\discrete{h}{\varepsilon}{P}$ executes to $\discrete{h}{\varepsilon}{P_1}$ with some state $\beta_1$. Denote $\dd (\sigma_1, \beta_1)$ by  $\vare_1$, suppose $\vare_1 < \vare$ is $M_1h$ for some $M_1$, we prove that with $\vare_1$ as the initial error, after the execution of $P_1$ and $\discrete{h}{\varepsilon}{P_1}$,  the global error (denoted by $\vare_2$) is $M_2 h$ for some constant $M_2$. As a consequence, there must exist sufficiently small $h$ such that the global error of $P$ is less than $\vare$. Notice that for the special case when $P_1$ is $P$, $\vare_1$ is 0, and the above fact implies the theorem. The proof is given by structural induction on $P_1$.
\begin{itemize}
   \item Case $P_1=\pskip$: the discretized process is  $\pskip$. Obviously $\vare_2=\vare_1$.
   \item Case $P_1=(x:=e)$:  the discretized process  is  $x:=e$, where $e$ is an expression of variables, thus can be written as a function application of form $f(x_1, \cdots, x_n)$, among which $x_1, \cdots, x_n$ denote the variables occurring in $e$. After the assignment, only the value of $x$ is changed. Thus, from the definition of $\dd$, we have the fact $\vare_2 = \max(\vare_1, |a_2 - a_1|)$, in which $a_1 = \sigma_1(e)$ and $a_2 = \beta_1(e)$ represent the value of $x$ after the assignment. From the definition of $e$, $a_2 = f(\beta_1(x_1), \cdots, \beta_1(x_n))$. For each $i=1, \cdots, n$, there exists $\delta_i$ such that $\beta_1(x_i) = \sigma_1(x_i) + \delta_i$ and $|\delta_i| \leq \vare_1$. By the Lagrange Mean Value Theorem, the following equation holds:
       \[a_2= f(\sigma_1(x_1), \cdots, \sigma_1(x_n)) + \sum_{i=1}^n \frac{\partial f}{\partial x_i}(\sigma_1(x_i) + \theta \delta_i) \delta_i\]
       where $\theta \in (0, 1)$. From the fact $f(\sigma_1(x_1), \cdots, \sigma_1(x_n)) = a_1$,
       \[|a_2 - a_1| = |\sum_{i=1}^n \frac{\partial f}{\partial x_i}(\sigma_1(x_i) + \theta \delta_i) \delta_i| \leq \vare_1\sum_{i=1}^n\max_{o\in (\sigma_1(x_i) - \vare_1, \sigma_1(x_i) + \vare_1)}|\frac{\partial f}{\partial x_i}(o)| \]
       $n$ is a constant, and $\frac{\partial f}{\partial x_i}$ is bounded in the interval $(\sigma_1(x_i) - \vare_1, \sigma_1(x_i) + \vare_1)$, thus $|a_2 - a_1|$ is bounded by a multiplication of $\vare_1$ with a bounded constant. $\vare_2$ is the maximum of $\vare_1$ and this upper bound of $|a_2 - a_1|$. The fact holds obviously.
   \item Case $P_1=\pwait\ d$:   the discretized process  is  $\pwait\ d$. Obviously $\vare_2 = \vare_1$.
   \item Case $P_1=ch?x$:  the discretized process is $ch?:=1; ch?x; ch?:=0$. Notice that the auxiliary readiness variable $ch?$ is added in the discretized process, however, it will not introduce errors. Thus, we only consider the error between the common variables, i.e. process variables, of $P_1$ and its discretization.  There are two cases for the transitions of $P_1$ and $\discrete{h}{\varepsilon}{P_1}$. The first case is waiting for some time units. For this case, if the waiting time is finite, then let  the time durations for both sides be the same, $\vare_2 = \vare_1$ holds obviously; if the waiting time is infinite, indicating that a deadlock occurs, $\vare_2 = \vare_1$ holds also. For the finite case, at some time,
        an event $ch?c$ occurs, where $c$ is the value received, and as a consequence, $x$ is assigned to $c$. For both sides, let the value received, denoted by $c_1$ and $c_2$ respectively, satisfy $|c_1-c_2| \leq \widehat{M}\vare_1$ for some constant $\widehat{M}$. As a result, after the performance of the events $ch?c_1$ and $ch?c_2$ respectively, $\vare_2 = \max\{\vare_1, \widehat{M}\vare_1\}$.
   \item Case $P_1=ch!e$: the discretized process is $ch!:=1; ch!e; ch!:=0$. Same to input, there are two cases for the transitions of $ch!e$. For the first case, let the time duration for both sides be the same, thus $\vare_2 = \vare_1$ obviously. For the second case, the events $ch!\sigma_1(e)$ and $ch!\beta_1(e)$ occur, and from the proof for assignment, there must exist a constant $\widehat{M}$ such that $|\beta_1(e) - \sigma_1(e)| < \widehat{M} \vare_1$ holds. No variable is changed as a consequence of an output, thus, after the communication, $\vare_2 = \vare_1$ still holds.

   \item Case $P_1 = Q; Q'$: the discretized process is $\discrete{h}{\varepsilon}{Q}; \discrete{h}{\varepsilon}{Q'}$. By induction hypothesis, assume the error after the execution of $Q$ with initial error $\vare_1$ is $\vare_m$, then $\vare_2=M_3 \vare_m$ and $\vare_m=M_4h$ for some constants $M_3, M_4$. $\vare_2 = M_3M_4 h$ holds.

       \item Case $P_1 = B \rightarrow Q$: the discretized process is $B\rightarrow \discrete{h}{\varepsilon}{Q}$. From the assumption $\dd (\sigma_1, \beta_1) = \vare_1=M_1h$, then there exists sufficiently small $h$ such that $\vare_1 < \vare$. Let $\vare < \epsilon$, then $\dd (\sigma_1, \beta_1) < \epsilon$.   $P$ is $(h, \epsilon)$-robustly safe, thus if $\sigma_1(B)$ is true, from the definition that the distance between $\sigma_1$ and any state that makes $\neg B$ true is larger than $\epsilon$, we can prove that $\beta_1(B)$ must be true. For this case, $Q$ and $\discrete{h}{\varepsilon}{Q}$ will be executed. By induction hypothesis, we have $\vare_2=M_3h$ for some constant $M_3$. Likewise, if $\sigma_1(B)$ is false, then $\beta_1(B)$ must be false. For this case, $P_1$ terminates immediately.  By induction hypothesis, $\vare_2=  \vare_1$, thus the fact holds.

       \item Case $P_1 = Q \sqcap Q'$: the discretized process is $\discrete{h}{\varepsilon}{Q} \sqcap \discrete{h}{\varepsilon}{Q'}$. There are two cases for the execution of both $P_1$ and its descretized process. By making the same choice, suppose $Q$ and $\discrete{h}{\varepsilon}{Q}$ are chosen to execute.
           By induction hypothesis, we have $\vare_2=M _3h$ for some constant $M_3$. The other case when $Q'$ and $\discrete{h}{\varepsilon}{Q'}$ are chosen can be proved similarly.

       \item Case $P_1 =  \evolutionn{\dot{\xx}=\ff(\xx)}{B}$: Let $X(t, \sigma_1(\xx))$ represent the trajectory of  $\dot{\xx}=\ff(\xx)$ with initial value $\sigma_1(\xx)$.  $\dot{\xx}=\ff(\xx)$ is GAS with the  equilibrium point $\bar{\xx}$, then $\lim_{t \rightarrow \infty}\| X(t, \sigma_1(\xx))\| = \bar{\xx}$. There must exist $T$ such that when $t>T$, $ \| X(t, \sigma_1(\xx)) - \bar{\xx}\| < \vare$.  Before time $T$, by Theorem~\ref{theorem:globalerror} and the initial error $\vare_1$ is $M_1h$, thus the global error between $X(t, \sigma_1(\xx))$ and the corresponding discretized $\xx_n$  is $M_3h$ for some constant $M_3$ defined in  Theorem~\ref{theorem:globalerror}.
           %Then there must exist sufficiently small $h$ such that the global error before $T$ is also less than $\vare$.
           Now we consider the escape of $P_1$ because of the failure of $B$, and how the discretized process behaves.

           There are two cases. First, if $B$ is always true along the trajectory $X(t, \sigma_1(\xx))$ on the infinite interval $[0, \infty)$, then $N(B, \vare)$ will be always true for the discretized points $\xx_n$. For this case, $P_1$ and the discretization will go close to the equilibrium point eventually, and the distance between them is always less than $\vare$ by choosing $h$ satisfying $M_3h < \vare$. Second,  if $B$ fails to hold for some $X(t_f, \sigma_1(\xx))$ at time $t_f$, then from the assumption that $P$ is $(\delta, \epsilon)$-robustly safe, there exists $\widehat{t}$ such that $\widehat{t} -t_f<\delta$ and for all $ \sigma$ satisfying $\dd(\sigma, \sigma_1[\xx\mapsto X(\widehat{t}, \sigma_1(\xx) )]) < \epsilon$, $\sigma \in N(\neg B, -\epsilon)$. Assume $\widehat{t} \in (t_N, t_{N+1}]$ for some $N$, so   $\dd(\sigma_1[\xx\mapsto X(\widehat{t}, \sigma_1(\xx) )], \beta_1[\xx \mapsto \xx_{N+1}])< M_3 h$. Let $h$ be sufficiently small such that $M_3 h < \epsilon$. Thus $\beta_1[\xx \mapsto \xx_{N+1}]\in(N(\neg B, -\epsilon))$, which implies $\beta_1[\xx \mapsto \xx_{N+1}](N(B, \vare))$ is false. As a result, the discretization of continuous evolution stops update correspondingly at time $t_{N+1}$. Thus we know that, the continuous evolution runs for $t_f -t_0$ time units in all and then escapes, and the discretization for $t_{N+1} - t_0$ time units, for the initial time $t_0$. Obviously the time precision $|t_{N+1} - t_f| \leq |t_{N+1} - \widehat{t}|+|\widehat{t} - t_f| \leq h + \delta \leq (1+\lceil \frac{\delta}{h}\rceil) h$ holds. Meanwhile, the value precision $\|X(t_f, \sigma_1(x)) - \xx_{N+1}\| \leq \|X(t_f, \sigma_1(\xx)) - X(\widehat{t}, \sigma_1(\xx))\| + \|X(\widehat{t}, \sigma_1(\xx)) - \xx_{N+1}\| \leq \max_{\xi \in [t_f, t_{N+1}]}\|{X'(\xi)} ( \\ t_f - t_{N+1})\| + M_3 h < \max_{\xi \in [t_f, t_{N+1}]}\|{X'(\xi)} (\delta + h)\| + M_3 h$. Thus by choosing $h$ sufficiently small, and with proper $\epsilon$, the error is less than $\vare$. The fact is thus proved for all the cases.

       %$ \fracN{\evolutionn{\dot{\xx}=\ff(\xx)}{B}}{(N(B,\vare) \rightarrow (x:=x+h f(x); \pwait \ h))^{\ulcorner\frac{T}{h}\urcorner}; N(B, \vare) \rightarrow (x:=x_0; \nstop)}$

        \item  Case $P_1 = \exempt{\evolutionn{\dot{\xx}=\ff(\xx)}{B}}{i\in I}{io_i}{Q_i}$:  First of all, notice that in the discretization of $P_1$, the auxiliary variables $io_i, \overline{io_i}$ are added for assisting the execution of interruption. These variables do not introduce errors.  Let $X(t, \sigma_1(\xx))$ represent the trajectory of  $\dot{\xx}=\ff(\xx)$ with initial value $\sigma_1(\xx)$.  $\dot{\xx}=f(\xx)$ is  GAS with the  equilibrium point $\bar{\xx}$, then $\lim_{t \rightarrow \infty}\| X(t, \sigma_1(\xx))\| = \bar{\xx}$. There must exist $T$ such that when $t>T$, $ \| X(t, \sigma_1(\xx)) - \bar{\xx}\| < \vare$. According to the transition semantics of communication interrupt, there are several cases. If the communications $\{io_i\}$ never occur, then the execution of the communication interrupt is equal to the execution of the continuous evolution. Correspondingly, in the discretized process, the first, the second and the fourth lines are executed depending on whether the continuous evolution terminates or not. Similar to the proof of the continuous evolution, the fact holds for this case. Otherwise, there must exist time $T_c$ such that for the first time some communications $\{jo_j\}$ for $j \in J \subseteq I$ get ready simultaneously, while others in $I\backslash J$ are not. Let $T$ be sufficiently large such that $T_c < T$ holds. For this case, before time $T_c$, $P_1$ executes by following $\dot{\xx}=\ff(\xx)$, and at time $T_c$, the external choice between $\{jo_j\}$ for $j \in J \subseteq I$ occurs and after the communication the corresponding $Q_j$ is followed for some $j$. Correspondingly, in the discretized process, the first and the third lines are executed.  Suppose $T_c \in (x_{nn}, x_{nn+1}]$ for some $nn$, i.e. the time $T_c$ occurs in the $(nn+1)$-th discretized interval. Because the variables $io_i$s and $\overline{io_i}$s do not introduce errors, plus their definitions, we know that at $x_{nn+1}$, $\exists i. io_i \wedge \overline{io_i}$ is detected to turn to true. Before the communication, for $P_1$, the time delay $T_c - t_0$ occurs, and for the discretized process, the time delay $x_{nn+1} - t_0$ occurs, where $t_0$ denotes the initial time of $P_1$. Obviously $|(T_c - t_0) - (x_{nn+1} - t_0)| < h$ holds. The global error is $M_3h$ for some constant $M_3$ before time $T_c$ obviously.  When the continuous evolution is interrupted, according to the definition of the discretized process, the fact also holds by induction hypothesis.

        \item Case $P_1 = Q^*$: the discretized process is $\discrete{h}{\varepsilon}{Q}^*$. From the definition of repetition, there exists a finite $N >0$ such that $P_1 = Q^m$ and $m \leq N$. Let in the discretized process the upper bound of the number of repetition be also $N$. Denote the global error after the $n$-th ($n \leq N$) execution of $Q$ by $\omega_{n}$.  By induction hypothesis on $Q$ and $\discrete{h}{\varepsilon}{Q}$, there exist $\widehat{M}_i$s ($i \in \{1, \cdots, N\}$) such that
             $\omega_i = \widehat{M}_i h$ for $i= \{1, 2, \cdots, N\}$. The fact is proved.

        \item Case $P_1 = Q\|Q'$: the discretized process is $\discrete{h}{\varepsilon}{Q} \| \discrete{h}{\varepsilon}{Q'}$. The global error is the maximum of the errors for the discretization of $Q$ and $Q'$. According to the transition semantics of HCSP, there are several cases for execution of parallel composition. If no compatible communication events over the common channels of $Q$ and $Q'$ exist, there are three cases: if $Q$ takes a $\tau$ event, then $\discrete{h}{\varepsilon}{Q}$ is able to take a same $\tau$ event, and vice versa, by induction hypothesis, the global error after the $\tau$ event is $M_3h$ for some constant $M_3$ obviously; the symmetric case when $Q'$ and $\discrete{h}{\varepsilon}{Q'}$ take a $\tau$ event can be handled similarly; for the third case, if both $Q$ and $Q'$ take progress for $d$ time units, then $\discrete{h}{\varepsilon}{Q}\| \discrete{h}{\varepsilon}{Q'}$ is able to take progress for $d$ time units, and vice versa, by induction hypothesis, the global error after the $d$ time duration is $M_4h$ for some constant $M_4$ obviously.
        For the case that a communication over a common channel $ch$ of $Q$ and $Q'$  occurs, according to the semantics, there is some value $c$ such that the events $ch?c$ and $ch!c$ occur for the two sides of $P_1$ respectively. Correspondingly, there is another value $c'$ such that the events $ch?c'$ and $ch!c'$ occur for the two sides of $\discrete{h}{\varepsilon}{P_1}$ respectively. Similar to the proof of assignment, we can prove that there must exist some constant $M_5$ such that $|c-c'| < M_5\vare_1$ holds. This fact is implied in the proofs of input and output events. After the occurrence of the communication, the global error is thus $M_6h$ for some constant $M_6$.
                 Finally,  three cases are left, $Q$ terminates earlier, or $Q'$ terminates earlier, or they terminate simultaneously. For all the cases, at the termination of $P_1$, the fact holds obviously by induction hypothesis.

\end{itemize}
Till now, we prove that the global error of the discretization, say $\vare_g$, is $Mh$ for some constant $M$, under the premise that $T_F$ for each ODE $F$ is chosen such that $F$ after time $T_F$ is $\vare$-closed to the corresponding equilibrium point. There must exist a sufficiently large $h$ such that $\vare_g < \vare$ holds. The fact is thus proved.
\QEDB